\documentclass[letterpaper,12pt]{article}
\usepackage{booktabs} % For formal tables
\usepackage[ruled]{algorithm2e} % For algorithms

\SetAlFnt{\small}
\SetAlCapFnt{\small}
\SetAlCapNameFnt{\small}
\SetAlCapHSkip{0pt}
\IncMargin{-\parindent}
\usepackage{chngcntr}
\usepackage[colorinlistoftodos,textsize=scriptsize]{todonotes}
\usepackage{dcolumn}
\usepackage[hang,flushmargin]{footmisc}

\renewcommand{\baselinestretch}{1.26}
\RequirePackage{geometry}
\RequirePackage{caption, float}
\usepackage{titlesec}
\titleformat{\section}{\bfseries}{\thesection}{1em}{\MakeUppercase}
\titleformat{\subsection}{\bfseries}{\thesubsection}{1em}{}
\titleformat{\subsubsection}{\itshape}{\thesubsubsection}{1em}{}
\titlespacing\section{0pt}{12pt plus 4pt minus 2pt}{4pt plus 2pt minus 2pt}
\titlespacing\subsection{0pt}{12pt plus 4pt minus 2pt}{4pt plus 2pt minus 2pt}
\titlespacing\subsubsection{0pt}{12pt plus 4pt minus 2pt}{4pt plus 2pt minus 2pt}

\usepackage{graphicx}
 
\usepackage{amsmath,bm,comment,amsfonts,amssymb,mathabx}
\allowdisplaybreaks
\usepackage{mathtools}
\usepackage{enumitem}
\setlist[enumerate]{leftmargin=1.5cm,rightmargin=0.5cm,noitemsep, topsep=2pt}
\usepackage{cancel}

\definecolor{npurple}{RGB}{82,0,99}
\definecolor[named]{ACMDarkBlue}{cmyk}{1,0.58,0,0.21}

\usepackage[unicode=true, pdfusetitle,
bookmarks=true,bookmarksnumbered=false,bookmarksopen=false,
breaklinks=false,pdfborder={0 0 1},backref=false,colorlinks=false]
{hyperref}
\hypersetup{
	colorlinks,
	linkcolor = npurple,
	citecolor = npurple,
  urlcolor=ACMDarkBlue,
  filecolor=ACMDarkBlue}

\usepackage[noabbrev]{cleveref}
\Crefname{equation}{}{}
\newtheorem{theorem}{Theorem}

\newtheorem{assumption}{Assumption}
\crefname{assumption}{assumption}{assumptions}

\newtheorem{claim}{Claim}
\crefname{claim}{claim}{claims}

\newtheorem{corollary}{Corollary}

\newtheorem{definition}{Definition}
\newtheorem{example}{Example}

\newtheorem{lemma}{Lemma}

\newtheorem{proposition}{Proposition}
\newtheorem{remark}{Remark}

\newenvironment{proof}{\textsc{Proof:}}{\hfill\ \rule{0.5em}{0.5em}\vspace{1em}}

\newcommand{\overbar}[1]{\mkern 1.5mu\overline{\mkern-1.5mu#1\mkern-1.5mu}\mkern 1.5mu}
\renewcommand{\d}{\mathrm{d}}
\newcommand{\R}{\mathbb{R}}
\newcommand{\E}{\mathbb{E}}

\newcommand{\F}{\mathcal{F}}
\newcommand{\M}{\mathcal{M}}

\def\K{\mathcal{K}}
\def\P{\mathbb{P}}
\def\S{\Omega}

\def\G{\mathcal{G}}
\def\tsum{{\textstyle{\sum}}}

\usepackage{natbib}

\usepackage{todonotes}
\usepackage{bm}
\usepackage{bold-extra}

\begin{document}

\title{\bf Robustly Optimal Mechanisms for Selling Multiple Goods}

\author{\textsc{Yeon-Koo Che}   \thinspace and \thinspace\ \textsc{Weijie Zhong}\thanks{Che: Department of Economics, Columbia University (email: yeonkooche@gmail.com); Zhong: Graduate School of Business, Stanford University (email: wayne.zhongwj@gmail.com)}}
\maketitle

\begin{abstract}  We study robustly optimal mechanisms for selling multiple items.  The seller maximizes revenue 
against a worst-case distribution of a buyer's valuations within a set of distributions, called an ``ambiguity'' set.  We identify the exact forms of robustly optimal selling mechanisms and the worst-case distributions when the ambiguity set satisfies various moment conditions on the values of subsets of goods.  The analysis reveals 
% also identify 
general properties of the ambiguity set that 
% lead to the robust optimality of 
justifies categorical bundling, which includes separate sales and pure bundling as special cases. 
\end{abstract}

\section{Introduction}

How should a seller sell multiple goods to a buyer?  The answer to this seemingly simple question remains elusive.  Unlike the single-good case,  the optimal selling mechanism for multiple goods is difficult to identify.   Traditional Bayesian approaches often lead to intricate mechanisms that are highly sensitive to the buyer's value distribution, as highlighted by \cite{daskalakis2017strong} and \cite{manelli2007multidimensional}. 
Moreover, simple mechanisms like item pricing or bundled pricing can significantly underperform compared to these theoretical optima (\cite{briest2010}; \cite{hart2013}), and it is challenging to identify when such simple approaches are effective (see \cite{daskalakis2017strong}, \cite{manelli2006bundling}).

Relaxing the Bayesian assumptions yields different optimal simple mechanisms in specific scenarios. \cite{carroll2017robustness} demonstrates that, when the seller knows only the marginal distributions of item values, separate sales at monopoly prices maximize worst-case revenue. 
 In a concurrent study, \cite{deb-roesler2021}
show that selling a grand bundle is optimal for worst-case revenue when the seller knows the value distribution but lacks information about the buyer's knowledge.

% ---is always maximized when  each good is sold separately at a monopoly price.
% \footnote{Optimal in the sense that it maximize the revenue against an adversarially chosen correlation structure.}  

% In practice, however, the pattern of product bundling is more complex. Typically, not all products a seller sells are sold separately or as a single bundle. Instead, a common practice is that groups of products are categorized into distinct bundles and sold separately.

In practice, product bundling exhibits greater complexity. Sellers rarely offer all products individually or as a single grand bundle. Instead, they commonly group products into distinct categories and sell them as (separate) bundles.\footnote{Also common is the bundling of complementary products such as computer hardware and software, hotels and flights, razors, and razor blades.  While our main model focuses on the additive value setting, we show in \Cref{cor:complementarity} that such complementarity can be easily incorporated into our model and is inconsequential for the prediction.}
Video streaming and cable TV services exemplify this, offering bundles of channels categorized by genre, such as news, entertainment, or sports. Similarly, investment banks and financial companies bundle assets into securities based on sectors like technology, energy, or healthcare.

  % Video streaming and cable TV services offer bundles of channels along several categories, such as news, entertainment, sports, etc.  Similarly, investment banks and financial companies offer securities that bundle many assets along several sectors, such as technology, energy, health, and the financial industry.  

In line with recent literature, this paper adopts a robustness approach by departing from Bayesian assumptions. Our distinctive contribution lies in characterizing the seller's knowledge structure that directly results in a specific pattern and scope of product bundling. Through this analysis, we provide unified insights into two extreme scenarios: full separation (selling all goods individually) and pure bundling (selling only a grand bundle). We achieve this by identifying the precise knowledge structures that underpin these contrasting bundling mechanisms.

 % Similar to the recent articles, the current paper takes the robustness approach by relaxing the Bayesian assumptions.  However, we characterize the seller's information structure that yields   a precise pattern and scope of bundling across goods. From this comprehensive analysis, we derive unified insights in two special cases by delineating the knowledge structures that underpin two contrasting bundling mechanisms---full separation and pure bundling.
 
The seller in our model has $n\ge 2$  heterogeneous goods to sell to a buyer.  The buyer has a quasilinear utility function additive in his valuations $ (v_1, ..., v_n) \in \mathbb{R}^n_+$  of the goods. The valuations are the buyer's private information,  unobserved by the seller.  In modeling a seller's information, we focus on the realistic scenario in which the seller accesses basic summary statistics such as means and variances of valuations along several categories of goods.  The goods are partitioned into {\it categories} $K\in \K$, where  $\K$ is an arbitrary partition, such that the seller knows only the means and some convex dispersion moments (to be defined) of each category value---the total valuation of the goods in each category.  The partition structure describes the \emph{granularity} of the seller's data, representing the smallest groups of goods on which the seller may infer about the group-level valuation. Meanwhile, the dispersion moments describe the \emph{precision} of the seller's inference of the summary statistics, generating ``confidence intervals'' for valuation estimates. 

We study the robustly optimal selling mechanism that maximizes the seller's expected revenue under the worst-case joint distribution consistent with the moment conditions.  For analysis, we consider a zero-sum game played by the seller seeking to maximize her revenue and the adversarial nature seeking to minimize it. 
The equilibrium of this game, or a saddle point,  identifies the optimal revenue guarantee for the seller.   

Our first main result, \Cref{thm:moment}, shows that the robustly-optimal mechanism consists of \textbf{$\K$-bundled sales}: {\it each bundle $K\in\K$ of items is sold separately at an independently distributed random price, or equivalently via a menu of lotteries with distinct prices.} As a direct corollary, a separate selling mechanism and a pure bundling mechanism are robustly optimal when $\K$ are the finest partition and the coarsest partition, respectively.  The intuition for \Cref{thm:moment} is that our seller faces three layers of uncertainties, and the $\K$-bundled sales mechanism hedges against each. First, the independent pricing of alternative categories guards the seller against the first layer of ambiguity concerning the correlation of their valuations across categories. Then, the bundling of each product category hedges against the second layer of ambiguity concerning how the value of each category is divided across its constituent items. Finally, the randomization of bundle prices, or a menu of prices for each bundle, hedges against the distributional ambiguity of each category value, consistent with moment conditions.\footnote{The intuition behind the first layer of uncertainty aligns with \cite{carroll2017robustness}, where each item represents a distinct category and only marginal distributions are known. Without the other two layers of ambiguity, this naturally leads to full separation as the optimal mechanism.  Meanwhile, the insights related to the second and third layers resonate with those provided by Deb and \cite{deb-roesler2021}. They demonstrate that similar uncertainties, arising from the buyer's information instead, justify pure bundling when the first layer of uncertainty is absent.}

Our second main result, \Cref{thm:necessity}, shows that the $\K$-bundling structure is not only sufficient but also necessary for the robust optimality of the mechanism: it is \emph{not} robustly optimal either to separate items within each product category  $K\in \K$ or to bundle multiple product categories in $\K$.  These other mechanisms may also be optimal against the distribution that justifies the use of $\K$-bundled sales. What fails them is robustness: they perform poorly against a different ``counterfactual'' distribution, the analysis of which reveals the seller's true motive for the use of $\K$-bundled sales. 
Specifically, a bundled sale of items within each category is motivated by the fear that a certain negative correlation across values would lead to revenue loss if items were sold separately. This finding harks back to the classic insight by  \cite{adams1976commodity}.  By contrast, separation across distinct categories is motivated by the fear of asymmetric distributions of the buyer's valuations across categories. If different categories are bundled, the bundle screens the buyer symmetrically across distinct categories, which results in screening inefficiency and revenue loss under the counterfactual distribution. 

%From this perspective, \cite{carroll2017robustness}'s model corresponds to the extreme case of fine-grained data: the seller's individual item-level knowledge leads to full separation and the perfect estimations of value distributions lead to deterministic prices for each item.  Our framework offers a natural generalization of his framework in which the seller's knowledge is much coarser in granularity and precision; most important, we are able to justify and characterize the use of a general categorical bundling of goods. 

We explore two extensions of our baseline model in \Cref{sec:info,sec:extensions}. The first extension adapts our framework to study ``informational ambiguity,'' where the seller has an unambiguous prior about the buyer's valuations but faces ambiguity with regard to the information the buyer has on his own valuations.
Such informational ambiguity is characterized by a collection of convex moment conditions, rendering our framework readily applicable.  We show that 
 if the valuation distribution is \emph{stochastically comonotonic}, i.e., they are obtainable by a suitable garbling of comonotonic distribution, then the worst-case signal that the buyer may have coincides with the worst-case distribution that justifies the use of pure bundling. 
 % diwhere each item's value constitutes a portion of a common component plus an idiosyncratic component, then the range of possible posterior distributions of buyer valuations can be effectively summarized by a properly-constructed statistic on the grand bundle. 
 Consequently, pure bundling is informationally robust. %Since stochastic comonotonicty includes exchangeability as a special case, the current result strictly generalizes \cite{deb-roesler2021} who establish the same result under exchangeable priors. 
 The second extension goes beyond the moment restrictions assumed in the baseline model and characterizes (more general) distributional restrictions, called \textbf{$\K$-Knightian ambiguity},  that justify the use of $\K$-bundled sales.  These two extensions demonstrate that the insights obtained from our analysis apply more broadly and resiliently beyond the specific settings considered in \Cref{sec:moment}.
 
 % In \Cref{sec:general}, we broaden the seller's knowledge about each product group $K$ from summary statistics to almost any distributional restrictions, termed \textbf{$\K$-Knightian ambiguity}. Under $\K-$Knightian ambiguity, the $\K-$bundled sales consistently prove to be robustly optimal.

 % affirms the resilience of our primary result. 

The current paper intersects with two broad strands of literature. First, it contributes to the multiproduct monopoly literature and, more broadly, the multidimensional screening and mechanism design literature. Representative works include \cite{mcafee1988multidimensional}, \cite{armstrong1996, armstrong1999}, \cite{manelli2006bundling, manelli2007multidimensional}, \cite{rochet/chone1988}, \cite{daskalakis2013mechanism, daskalakis2017strong}, \cite{hart2015maximal, hart2019better}, \cite{menicucci2015}, and \cite{haghpanah2020}.  The current paper departs from this literature by taking a robustness approach. 

Second, the current paper contributes to the literature on robust mechanism design.  Many authors study optimal mechanisms under the worst-case distribution of states. To the best of our knowledge, \cite{scarf1957minmax} was the first to adopt this approach in inventory management. \cite{carroll2015linearcontracts, carroll2019infoscore} apply the approach to contracting settings. \cite{bergemannschlag2008pricing} and \cite{ carrasco2018optimalselling} solve the single-item monopoly problem with neighborhood restrictions and moment conditions, respectively.  \cite{distributionallyRobust2019}, \cite{brooks2021maxmin}, \cite{He2019RobustlyOptimal} and \cite{che2022robustly} extend the framework to the multi-buyer auction setting, but still with one item.   As already discussed, \cite{carroll2017robustness} applies the robust mechanism design approach to a multi-item sale problem with known marginals, making it the closest antecedent of the current paper. We develop his framework further and provide a robustness-based rationale for general forms of categorical bundling, which include separate sales and pure bundling as special cases.\footnote{ Although  worst-case revenue maximization is a natural way to extend the standard Bayesian framework, %\footnote{The maximin approach has a long tradition in economics starting with \cite{wald1950} and has a well-known axiomatic foundation via ambiguity aversion (see \cite{gilboa-schmeidler1989}).} 
several authors have also considered other notions of robustness in mechanism design.  \cite{bergemann2008pricing}, \cite{guo2019robust} and \cite{koccyiugit2018robust} study the minimization of regret---namely, a revenue shortfall of the chosen mechanism relative to the complete-information optimal mechanism.  In particular, \cite{koccyiugit2018robust} finds a regret-minimizing mechanism for selling multiple items with known means and rectangular domain, which parallels the case treated in \Cref{app:domain}. Another objective popular in algorithmic mechanism design is the revenue ratio of simple mechanisms (often separate sales and pure bundling) to all mechanisms across all or a restricted set of valuation distributions. As the number of items grows large, the ratio tends to zero when the distributed is unrestricted (\cite{briest2010,hart2013}) and is bounded away from zero when item values are independently distributed (\cite{babaioff2014simple,hart2012approximate,li2013revenue}).}  In addition, to our knowledge, our necessity result of a robustly optimal mechanism is new to the literature.

Recent authors have also studied the optimal mechanism in the worst-case scenario in terms of the information possessed by agents; see \cite{du2018commonvalue},  \cite{bbm2016informationally}, \cite{brooks2021maxmin, brooksdu2019, brooks2021structure}, and \cite{deb-roesler2021}.  These papers assume that a seller is ambiguous about the buyer's information regarding the values of items and chooses an optimal mechanism robust with respect to the buyer's information.\footnote{\cite{brooks2021maxmin} employs both distributional and informational uncertainties.}  Such a model can be seen as a robust mechanism design problem in which the ambiguity set is determined by the seller's prior belief in a particular way. Among them, \cite{deb-roesler2021} deals with the multi-item selling problem. They show that pure bundling is informationally robust when the prior belief is exchangeable across alternative items. We identify a more general stochastic comonotonicity condition for the informational robust optimality of pure bundling, which nests the exchangeable prior assumption as a special case but also permits highly asymmetric priors.\footnote{\cite{deb-roesler2021} also explores the buyer-optimal information structure facing a Bayesian seller, which is not the main focus of our paper.}
%\footnote{ Although the current paper is concurrent with \cite{deb-roesler2021}, \Cref{sec:info} is subsequent to \cite{deb-roesler2021}.  Our generalization in \Cref{sec:info} should therefore be regarded as an extension of their contribution.} 
\cite{brooks2021structure} proves that in multi-item auctions, informational ambiguity implies the robust optimality of indirect mechanisms with one-dimensional message space. 

%While it remains an open question whether such one-dimensional screening mechanisms are robustly optimal under our distributional ambiguity setting, \wz{in \Cref{sec:necessity}, our methodology allowed us to exclude certain canonical types of one-dimensional mechanisms.}

The rest of the paper is organized as follows. \Cref{sec:model} introduces a model of multi-item sale and defines a notion of robust optimality.  \Cref{sec:moment} considers an ambiguity set defined by a combination of moment conditions and establishes the robust optimality of $\K-$bundled sales, which specialize to  separate sales and pure bundling when $\K$ are the finest and coarsest partitions, respectively.
%\Cref{ssec:separable}  and \Cref{ssec:pure-mean} consider ambiguity sets defined by a combinations of moment conditions that lead to the robust optimality of separate and bundled mechanisms, respectively. 
\Cref{sec:necessity} establishes that the main qualitative features of $\K$-bundled sales are necessary. \Cref{sec:info,sec:extensions} extend the optimality of $\K$-bundled sales to a setting with informational ambiguity and a setting with more general distributional ambiguity, respectively.  \Cref{sec:conclusion} concludes.

\section{Model} \label{sec:model}
A seller wishes to sell $n$ items to a single buyer.  The buyer has values 
$\bm{v}:=(v_1, ..., v_n)$ for the items whose distribution is unknown to the seller.\footnote{To be precise, the ``values'' $\bm{v}:=(v_1, ..., v_n)$ need not be true values but rather the estimates the buyer assigns to items.  In this sense, the ambiguity the seller faces is ultimately an informational one, arising from her ignorance of what the buyer ``knows.''}  The seller simply knows that the distribution lies within some {\it ambiguity set}  $\F\subset \Delta(\R^n_+)$ defined by a set of moment conditions.  

\paragraph{Moment conditions:}
The moment conditions are defined in terms of means and dispersions on a joint distribution. To define them, fix any joint distribution $F\in \Delta(\mathbb{R}_+^n)$.  First, we assume the seller has some knowledge about the means of item values. Given $F$, let $\mu_i(F):=\mathbb{E}_F[v_i]$ denote the mean value of item $i$. Next, the seller has some knowledge about the dispersion of values of arbitrary subsets of items.  Specifically, let $\K$ be an arbitrary partition of the goods, with its element $K\in \K$ interpreted as a \emph{product category}.  For each product category $K\in \K$, we let 
 $$\sigma_K(F) :=\E_F\left[\phi_K\left(\tsum_{i\in K}v_i\right)\right]$$
 be the \emph{dispersion} of category $K$'s value under $F$, 
 where $\phi_K: \R_+\to \R_+$ is a twice-differentiable convex function satisfying $\phi_K''\ge \varepsilon$ for some $\varepsilon>0$.   We will refer to such a function as \emph{convex moment function}.  

We assume that the seller knows item value means and dispersion lie in some arbitrary nonempty convex and compact set $\S\subset \R_+^{n+|\K|}$.\footnote{Given the linearity of a mean, our ambiguity set allows a mean condition to apply to the value of each category, instead of the value of each item. As is clear from the proof of \Cref{thm:moment}, this has no effect for the qualitative features of the robustly optimal mechanism.}  Formally, the seller faces an ambiguity set
\begin{align} \label{eq:ambiguity-partial}
\F:=\Big\{F\in\Delta(\mathbb{R}_+^n): \left(\mu_i(F),\sigma_K(F)\right)_{i\in N,K\in \K}\in \S\Big\}.
\end{align}

An example, which will be used throughout,  illustrates the nature of ambiguity facing the seller. 
\begin{example}\label{eg:1}    Suppose the seller has $3$ goods with $N=\{1,2,3\}$.  The seller does not know the distribution of the items' valuations to the buyer except that she knows the means $\E[v_1]= 0.5,\ \E[v_2]=\E[v_3]=0.3$, and variances $\mathbb{V}[v_1]=\mathbb{V}[v_2+v_3]=0.1$ for two categories, given by the partition  $\K=\{\{1\},\{2,3\}\}$. In this case,   the convex moment functions are  $\phi_{\{1\}}(v)=\phi_{\{2,3\}}(v)=v^2$, and the moment constraint set is $\S=\left\{(0.5,0.3,0.3,0.35,0.46)\right\}$.
% \footnote{\wz{[I suggest remove this fn given fn 7. ]}In the example, we assume that the seller knows $\E[v_2]$ and $\E[v_3]$ rather than $\E[v_2+v_3]$. This allows us to better visualize the analysis.   As the proof of \Cref{thm:moment} suggests, given the knowledge of means of a \wz{category} , a further knowledge about \wz{individual means} is inconsequential.} 
\end{example}

 \paragraph{Feasible mechanisms:}
The seller is free to choose any selling mechanism. By the revelation principle, it is without loss to focus on direct revelation mechanisms, denoted by $M=(q(\bm{v}),t(\bm{v}))$, where the \emph{allocation rule} $q:\bm{v}\mapsto[0,1]^n$ specifies the probability of allocating each item to the buyer, and the \emph{payment rule}  $t:\bm{v}\mapsto\R^+$ specifies the expected payment received from the buyer, both as Borel measurable functions of the vector $\bm{v}$ of values reported by the buyer.
% \footnote{\wz{[R1 requests removing this fn]} Note that ambiguity in our model does not invalidate the revelation principle. Interpreting $(q(v), t(v))$ as the outcome when the buyer has value $v$, the feasibility conditions  $(IC)$ and $(IR)$ are necessary and sufficient for the outcome to be implementable.}   
The mechanism  satisfies {\it incentive compatibility} and {\it individual rationality}:  
\begin{align*}
	&\bm{v}\cdot q(\bm{v})-t(\bm{v})\ge\sup_{\bm{v}'\in\R^n_+}\bm{v}\cdot q(\bm{v}')-t(\bm{v}')\tag{IC}\\
	&\bm{v}\cdot q(\bm{v})-t(\bm{v})\ge 0\tag{IR}
\end{align*}
for  each $\bm{v}\in \R^n_+$. Let $\M$ denote the set of all direct mechanisms satisfying the $(IC)$ and $(IR)$ constraints---called {\it feasible} mechanisms.\footnote{\label{fn:mechanism} It is without loss to require $(IC)$ and $(IR)$ for all  types in $\mathbb{R}_+^n$, rather than only for $\bm{v}\in \bigcup_{F\in \mathcal{F}}\text{supp}(F)$.  \Cref{prop:mechanism:extension} shows that, for any feasible mechanism defined on $\bigcup_{F\in \mathcal{F}}\text{supp}(F)$, one can find a Borel measurable extension that satisfies  $(IC)$ and $(IR)$  for all types in  $\R^n_+$ and implements the same outcome for the types in the original domain.} 

\textbf{$\K$-bundled sales}: Among the feasible mechanisms, certain types of mechanisms will be of special interest to us. Consider the partition $\K$. 
% Each element of the partition can be interpreted as a bundle of items; the partition $\K$ then represents a particular collection of bundles.
The seller may bundle each category  $K\in \K$ of items and sell that bundle separately from the other categories of items. For each item $i$, let $K(i)$ be the category $K\in \K$ containing it.  Formally, we say a feasible mechanism $M:=(q,t)\in \M$ is a  \emph{$\K$-bundled sales} mechanism if, for each $K\in \K$, there exists a feasible (one-dimensional) mechanism $q_K: \R_+ \to [0,1]$ and $t_K:\R_+ \to \R$ such that $t(\bm{v})=\sum_{K\in\K} t_K(\sum_{j\in K} v_j)$ and  $q_i(\bm{v})=q_{K(i)}(\sum_{j\in K} v_j)$.%\footnote{\yk{[Do we  need this  footnote? The definition of $\K$-bundled sales is orthogonal to their feasibility; i.e., we didn't claim them to be feasible.]}  The feasibility of these one-dimensional mechanisms is implied by the feasibility of $M=(q,t)$.  Specifically, $(IC)$ of $(q_K,t_K)$ is implied by $(IC)$ of $(q,t)$. Likewise, $(IR)$ of $(q,t)$ implies that $t_K$ can be adjusted by a constant so that $(q_K,t_K)$ satisfies $(IR)$.} 
That is, the mechanism sells each bundle $K$ with probability $q_K$ and collects expected payment $t_K$. Let $\M_{\K}$ denote the set of all feasible $\K$-bundled sales mechanisms.

In our leading example,  the $\K$-bundled sales mechanism involves two bundles: the first bundle is good 1 only and the second bundle comprises goods 2 and 3, priced independently according to distributions $g_{1}$ and $g_{23}$.   
    
$\K$-bundled sales include two canonical mechanisms as special cases.  When $\K$ is the finest partition, namely when $\K=\{\{1\},..., \{n\}\}$, $\K$-bundled sales reduce to selling each item separately; we will refer to this as a \emph{separate sales} mechanism.  When $\K$ is the coarsest partition, namely when $\K=\{\{1,..., n\}\}$, 
 $\K$-bundled sales reduces to selling all items as a single grand bundle; we will call such a mechanism \emph{pure bundling}.

\paragraph{Robustness solution concept:}   The seller's revenue from a mechanism $M\in \M$ given value distribution $F$ is $R(M,F):=\int t(\bm{v}) F(\d \bm{v})$.\footnote{Here, we implicitly assume that the seller values each item at zero.  This is without loss. If there are unit costs $\bm{c}=(c_i)\ge 0$ for the items, then the problem facing the seller is exactly the same as in our model in which she faces $\bm{w}=\bm{v}-\bm{c}$ as the buyer's valuations and zero costs.  Robustly optimal mechanisms are then obtained upon an appropriate change of variables.  Specifically, a saddle point $(M^*, F^*)$  in our original model without costs remains a saddle point in terms of $\bm{w}$ in the new model.} 
 %The worst revenue from that mechanism is $\inf_{F\in \mathcal{F}}R(M,F)$.
 Let $R\in \R$ be a {\it revenue guarantee} if there exists a mechanism $M\in \M$ such that 
$R(M,F)\ge R$ for all $F\in \F$.   The seller's objective is to maximize the revenue guarantee.  Let 
$$R^*:=\sup_{M\in \M}\inf_{F\in\F} R(M,F)$$
be the \emph{optimal revenue guarantee}, and we say the mechanism attaining $R^*$ is  \emph{robustly optimal}. 

\subsection{Discussions of Model}\label{ssec:discussions}

The main elements of the model are motivated as follows.

 \paragraph{Partitional  knowledge  structure:}
As illustrated in the introduction, the partition of goods $\K$ describes the granularity of the seller's knowledge about the value distribution. In practice, the partitional structure of knowledge often reflects the intrinsic characteristics of goods. For instance, a financial broker may categorize stocks according to the sector possibly because stocks within a sector are influenced by common factors, making them more comparable to each other than to stocks from different sectors.  Similarly, in media and entertainment, products such as movies, music, and books are grouped by genres, acknowledging the shared preferences of consumers for content within the same genre. Wholesale distributors who supply many goods to retail grocers categorize goods according to their industry, brand, and grade. Market research then employs this categorization to derive key summary statistics, such as the mean and variance of consumer valuations for these product clusters, resulting in a partitional knowledge structure. 
% Such categorization increases the statistical power of the statistics as well as simplifies the messages that agents need to communicate.   

% This systematic approach to data collection is not merely a reflection of the products' inherent similarities but also a practical response to the complexities of market analysis. By categorizing products, researchers can more effectively analyze the valuations. For instance, a digital news provider conducting A/B tests via advertising promotion might find individual article promotion inefficient due to the small sample size, leading to statistically underpowered results. By instead focusing on a group of articles related by topic, the provider can leverage a larger viewer base, enhancing the statistical power of the tests. This approach capitalizes on the informational synergies among similar goods. In essence, the collection of data on categories of goods with shared characteristics exploits these informational complementarities, maximizing the efficiency and effectiveness of market research.

% \wz{\cite{carroll2017robustness}'s model of known marginal distributions corresponds to the extreme case where $\K$ is the finest partition, i.e. the seller's data granularity is the finest. In this case, the corresponding $\K$-bundled sales mechanisms reduce to separate sales, as are predicted to be robustly optimal by \cite{carroll2017robustness}.\footnote{\cite{carroll2017robustness}'s model is not nested by our main model with moment conditions, but is nested by our generalization in \Cref{sec:extensions}. }}  

\paragraph{$\K-$bundled sales in practice:} Besides the two extreme cases (fully separation and pure bundling), bundled sales with a non-degenerate partition structure are widely adopted in practice. Intriguingly, they are commonly adopted in the examples we introduced where the seller's knowledge exhibits the partitional structure. Content providers often sell bundled subscriptions based on the genres of the content. For example, the Wall Street Journal offers their professional clients six bundles of news partitioned based on the industry.\footnote{See \url{https://wsjpro.com/}.} Similarly, cable TVs and video streaming services typically categorize their channels into three bundles news, sports, and movies. In the supply chain of groceries, a wholesale distributor often offers a menu consisting of \emph{product lines}, each of which is a bundle of commodities of a specific brand. 

\paragraph{Moment conditions:} We consider the ambiguity set defined by moment conditions because moments (such as mean and variance) are the most natural information the seller may receive in practice. However, our assumption on $\F$ covers more than conventional moment conditions. 

 First, we consider a general convex moment function $\phi$, instead of conventional power moment functions (e.g. \cite{carrasco2018optimalselling}).  This generality enables us to extend our results to other types of ambiguities. In \Cref{sec:info}, we demonstrate that the so-called \emph{informational ambiguity}---the seller being ambiguous about what the buyer knows about---can be represented as a dispersion condition corresponding to a particular convex moment function. Additionally, in \Cref{app:domain}, we show that the case of {\it domain restriction}, where the value of each bundle $K\in \K$,  $\sum_{i \in K} v_i$,  lies in some interval $[0,\overbar{v}_K]$, can also be analyzed as a limiting case of convex moments. 

 Second, the generality of the ``confidence set'' $ \S$ allows us to capture a wide range of scenarios regarding the seller's ambiguity.  
%For instance,   $\S$ could be arbitrarily close to $\R_+^{n+|\K|}$. In this case, the moment condition entails almost no restriction on the ambiguity set.\footnote{A truly unrestricted ambiguity set is uninteresting, however, since the worst case distribution for the seller will be degenerate at zero. }
 For instance, $\S$ could be arbitrarily close to $\R^{|\K|}$ when projected to the last $|\K|$ dimensions, in which case the convex moments would be unrestricted.\footnote{In this special case, the ambiguity set is compatible with any arbitrary partition $\K'$. Our \Cref{thm:moment} would then imply that all bundling structures (including pure bundling and full separation) are robustly optimal.}
At the other extreme,  $\S$ could be a singleton; the seller then knows the exact means and dispersion of individual item values. As another example, $\S$ could be characterized by a system of inequalities:  $\psi_j(\mu_1(F), ..., \mu_n(F))\ge 0,$ for some concave functions $\psi_j, j=1, ..., n$. This allows for cases in which the seller knows the average values of subsets of items.\footnote{For instance, we could have  $\sum_{i\in K}\E_F[ v_i] =  m_K$ for each $K\in \K$.} Generally, one can view the size of $\S$ as capturing the magnitude of ambiguity the seller faces concerning the relevant bundle values.

\paragraph{Robustness solution concept:} Throughout the paper (except \Cref{sec:extensions}) we find the robustly optimal mechanism by solving a simultaneous-move zero-sum game.
The equilibrium $(M^*, F^*)\in (\M, \F)$ of the zero-sum game is a \emph{saddle point}; i.e., $\forall M\in \M$, $\forall F\in \F$,
\begin{align}
	R(M,F^*)\le R(M^*,F^*)\le R(M^*,F). \label{eqn:1}
\end{align}
It is well-known that a saddle point gives rise to an optimal revenue guarantee (see \cite{osborne-rubinstein1994}, Proposition 22.2-b) and the guarantee does not change with the order of moves:
\begin{align} 
R(M^*,F^*)=\max_{M\in\M}\min_{F\in\F}R(M,F)=\min_{F\in\F}\max_{M\in\M}R(M,F)=R^*. \label{eqn:minmax}
\end{align}

\section{Robust Optimality of $\K$-bundled sales} \label{sec:moment}

In this section, we solve for the robustly optimal selling mechanism for the seller who faces the ambiguity set indexed by an arbitrary partition $\K$ and the set $\S$ of possible moments.  We begin with the main theorem:

% In \Cref{sec:moment}, we show that $\K$-bundled sales mechanism is robustly optimal. In \Cref{sec:necessity}, we establish that the $\K$-bundling structure is not only sufficient but also necessary for robust optimality.

\begin{theorem}	\label{thm:moment}
It is robustly optimal for the seller to use a $\K$-bundled sales mechanism.
\end{theorem}
\begin{proof}
    See \Cref{Proof:moment}.
\end{proof}

To prove \Cref{thm:moment}, we construct a $\K$-bundled sales mechanism $M^*\in \M$ together with the distribution $F^*\in \F$ such that they form mutual best responses.   Here, we illustrate the main economic intuition behind the construction using \Cref{eg:1};  the general construction of $(F^*, M^*)$ and proof appear in \Cref{Proof:moment}.

\paragraph{Construction of $F^*$:}  Recall that in \Cref{eg:1}, $\K=\{\{1\},\{2,3\}\}$, with the ambiguity set $\F$ constrained in terms of means and variances of valuations along two categories, $\{1\}$ and $\{2,3\}$, of goods.   Nature chooses $F^*\in \F$ that has a one-dimensional support depicted in \Cref{fig:F1}: the support forms a ray emanating from the origin and contains a vertical segment toward the end. Specifically, nature draws a one-dimensional random variable $X$ distributed from $[1,\infty)$ according to cdf $$H(x):=1-1/x,$$
(i.e. Pareto distribution).  The valuation of each good $i=1,2,3$ is then determined as
$$v_i=\min\{\alpha_i X, \beta_i\},$$
where $(\alpha_i,\beta_i)$'s are chosen (uniquely) to satisfy the moment conditions (i.e., $\E[v_1]= 0.5,\ \E[v_2]=\E[v_3]=0.3$, and variances $\mathbb{V}[v_1]=\mathbb{V}[v_2+v_3]=0.1$).  

\begin{figure}[htbp]
	\centering
	\includegraphics[width=0.5\textwidth]{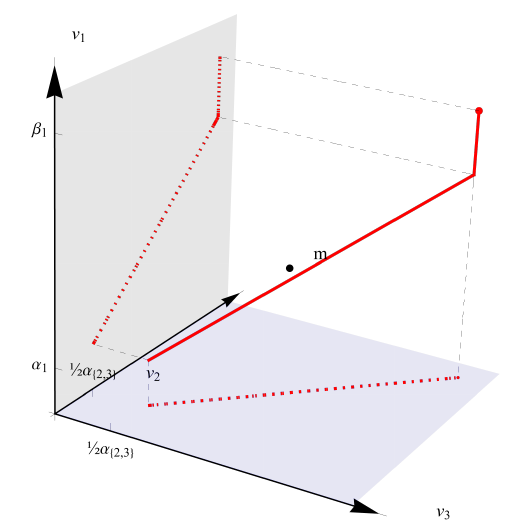}
	\caption{\small The solid red curve is the support of the joint distribution $F^*$. The dashed curves are its projections to the subspaces.}
	\label{fig:F1}
\end{figure}

There are two properties notable about $F^*$.  The first is that the support is comonotonic within each category;  see, for example in \Cref{fig:F1}, its projection onto the $(v_2,v_3)$ space. As will be explained in \Cref{sec:necessity}, this property incentivizes the seller to bundle each category of goods. Second, $F^*$ is designed to suppress the seller's revenue as much as possible. To see this, note first that, since $F^*$ is one-dimensional, the seller's mechanism design problem reduces to a standard one-dimensional screening problem (with multi-dimensional allocations). The virtual value of each item (as a function of $x$) is 
\begin{align*}
    J_i(x)=V_i(x)-V_i'(x)\frac{1-H(x)}{h(x)}=
    \begin{cases}
        0 &\text{if } x<\beta_i/\alpha_i\\
        \beta_i&\text{if } x\ge \beta_i/\alpha_i
    \end{cases}.
\end{align*}

  In words, nature ``levels'' the seller's virtual valuation.  The ``flat'' virtual value function means that all mechanisms that allocate each item $i$ to buyer type with valuation $\beta_i$ are optimal. Since any $\K-$bundled sales mechanism that sets bundle $\{1\}$'s price  in $[\alpha_1,\beta_1)$ and bundle $\{2,3\}$'s price in $[\alpha_2+\alpha_3, \beta_2+\beta_3)$, respectively, satisfy this property, $M^*$ constructed below is optimal given $F^*$. The resulting revenue is precisely what the seller receives by charging the lowest prices in the support of $F^*$, i.e., $\alpha_1+\alpha_2+\alpha_3$.\footnote{This property of the Pareto distribution has been exploited in other papers, \cite{carrasco2018optimalselling} and \cite{roesler2017buyer}, in the single-good case.  Unlike the latter paper, the seller does not choose the lowest price in the support of the worst-case distribution but instead mixes over the support to guarantee revenue.  The role of  the Pareto distribution here is therefore to keep the seller's revenue at the level she would enjoy by setting the lowest prices in the support. }

% This feature of $H$ incentivizes the buyer to set the lowest possible prices without violating the moment conditions.  
% Specifically, facing $F^*$, the seller's best response is price two categories, $\{1\}$ and $\{2,3\}$, at their respective infimum valuations $\alpha_1$ and $\alpha_2+\alpha_3$.  As will be seen, they determine the revenue for the seller. 

% \wz{[We did not explain why $M^*$ is optimal, as is pointed out by R3] Intuitively, nature ``levels'' the seller's virtual valuation. This feature of $H$ incentivizes the buyer to set the lowest possible prices without violating the moment conditions. Note that this include a special case where bundles $\{1\}$ and $\{2,3\}$ are priced at their respective infimum valuations $\alpha_1$ and $\alpha_2+\alpha_3$, which determines optimal revenue guarantee for the seller.\footnote{These properties of the Pareto distribution has been exploited in other papers, \cite{roesler2017buyer} and \cite{carrasco2018optimalselling},  in the single-good case. Specifically, \cite{roesler2017buyer} use the latter property to incentivize a profit-maximizing seller to subsequently set a price that maximizes consumer surplus.}} 

% In other words, any mechanism that allocates item $i$ when $v_i=\beta_i$ with probability $1$ is optimal, which is satisfied by $M^*$. Intuitively, nature's worst-case distribution also maximally ``hedges'' by leveling the seller's virtual value function. 

\textbf{Mechanism $M^*$}:   The optimal $\K$-bundling mechanism sells the two categories $\{1\}$ and $\{2,3\}$ as two separate bundles, and prices the respective bundles independently and randomly according to density functions:
  $$g_{1}(p)=2\lambda_{1}\cdot \frac{\beta_{1}-p}{p}$$ on support $[\alpha_{1},\beta_{1}]$, and 
  $$g_{23}(p)=2\lambda_{23}\cdot \frac{\beta_{23}-p}{p}$$
  on  support $[\alpha_{23},\beta_{23}]$, where
 $\alpha_{23}=\alpha_2+\alpha_3$ and $\beta_{23}=\beta_2+\beta_3$, and $\bm{\lambda}=(\lambda_{1},\lambda_{23})$ are scale factors chosen so that the densities integrate to ones.  

Intuitively, the $\K$-bundling mechanism is designed to ``hedge'' against possible deviation by nature away from $F^*$. 
To see this, let us recall what nature is capable of doing. First, nature can freely shift values across items within the bundle $\{2,3\}$, as the moment condition only constrains its total value.  Second, nature can freely correlate the values of categories $\{1\}$ and  $\{2,3\}$. Third, nature can redistribute the value of each bundle---bundle 1 and bundle $\{2,3\}$---while maintaining the same dispersion measures.  
   
Our mechanism makes it unprofitable for nature to deviate along these three channels. 
% balances nature's potential advantages from a mean-preserving redistribution of values via three main characteristics. 
First, the bundling of items 2 and 3 makes their valuations perfectly substitutable from the revenue standpoint; hence, nature has no incentive to redistribute values within the $\{2,3\}$ bundle.  Second,
the independent pricing of the two bundles means that nature never gains from manipulating the correlations of values across the bundles.
Finally,  the prices of the two bundles are randomized according to the particular hyperbolic forms of the densities precisely to eliminate any incentive by nature to redistribute value within each bundle.

 The construction and the argument generalize naturally to an arbitrary $\K$, leading to a pair of $\K$-bundled sales mechanisms $M^*$ and a one-dimensional valuation distribution $F^*$ that constitute a saddle point.  Two questions remain. First, how do we pin down the parameters $\bm{\alpha},\bm{\beta}$ and $\bm{\lambda}$, especially under the general set $\S$?  Second, why are different mechanisms \emph{not} robustly optimal? We relegate the answer to the first question in the formal proof in \Cref{Proof:moment}. The second question is addressed in \Cref{sec:necessity}.

\subsection{Implications of \Cref{thm:moment}}

\Cref{thm:moment} implies a few immediate corollaries. First, we provide a rationale for the use of both separate sales and pure bundling as special cases.
\begin{corollary}  \label{cor:special} If the seller faces the ambiguity set in \cref{eq:ambiguity-partial} where $\K$ is the finest partition, then separate sales of individual items are robustly optimal. If the seller faces ambiguity set in \cref{eq:ambiguity-partial} where $\K$ is the coarsest partition, then a sale of the grand bundle is  robustly optimal.
\end{corollary}

% More generally,  \Cref{thm:moment} rationalizes a   form of partial bundling that is ``aligned'' with the structure of the seller's dispersional knowledge, as represented by the partition $\K$. The rough intuition is as follows.  On one hand, the ambiguity about the correlation between values of bundles in $\K$  leads to separate sales of these alternative bundles. On the other, the ambiguity about how a given value of each bundle $K\in \K$ is distributed across items within  $K$ leads to the bundling of the items within $K$.  This insight will be further generalized in the next section.  

The next corollary expands the applicability of \Cref{thm:moment} beyond  the ambiguity set $\F$ in \Cref{eq:ambiguity-partial}.  
 
\begin{corollary} \label{cor:beyond} Suppose $(M^*,F^*)$ is a saddle point given an ambiguity set $\F$.  If $\widetilde \F\subset \F$ such that $F^*\in \widetilde \F$, then, $(M^*, F^*)$ is a saddle point given the  ambiguity set $\widetilde \F$.
\end{corollary}
\begin{proof}  The result follows since
	$R(M^*, F^*) \le R(M^*, F)$ for any $F\in \widetilde \F\subset \F$.
\end{proof}

This corollary states that  $F^*$ remains robustly optimal within any ambiguity set $\widetilde \F$ if it is in turn a subset of $\F$ defined in \Cref{eq:ambiguity-partial}. This simple corollary,  reminiscent of a revealed preference argument, turns out to be quite useful. For instance, one may find it plausible that item values are positively correlated so that the correlation coefficient between any pair of item values exceeds some number $\theta\in [0,1)$.  Since $F^*$ exhibits high correlation across all the $v_i$'s (they are perfectly correlated in the interior support),  it will satisfy this additional restriction for $\theta$ small enough, so one may conclude that the mechanism identified in \Cref{thm:moment} continues to be robustly optimal given the correlation condition. 

Next, we can extend \Cref{thm:moment} to allow for the types of complementarities considered by   \cite{deb-roesler2021}.
\begin{corollary}\label{cor:complementarity}
    Suppose the buyer's value of each subset of items, $J$ (not necessarily an element of $\K$), is given by $u_J\cdot \sum_{i\in J}v_i$, for some $u_J\in[0,1]$, and exhibits \emph{complementarity}: $u_J=1$ if $J\in \K$. Then, a $\K$-bundled sales mechanism is robustly optimal.
\end{corollary}
\begin{proof}
    See \Cref{proof:complementarity}.
\end{proof}

In the corollary, the value of all bundles except for those in $\K$ can be discounted by a factor of $u_J\le 1$ and $\K$-bundled sales remains robustly optimal. This generalization is straightforward since the revenue from $\K$-bundled sales is not affected by such complementarity, while all other mechanisms only under-perform.

How does the seller's revenue guarantee vary with 
the seller's knowledge?  Such comparative statistics is interesting in their own right since these parameters may be ex-ante controlled via market research efforts.  For instance, one may increase the precision of market forecast by reducing the dispersion; it would be interesting to know how such an investment may improve the revenue guarantee.  
Thanks to our complete characterization, such comparative statics can be readily performed.  

\begin{corollary}\label{cor:comp:statics}  Suppose \(\Omega=\{(\bm{m},\bm{s})\}\). Let $R^*(\bm{m},\bm{s})$ denote the optimal revenue guarantee as a function of $(\bm{m},\bm{s})$. Let $m_K=\sum_{i\in K}m_i$.
\begin{align}
    \begin{dcases}
        \frac{\d R^*(\bm{m},\bm{s})}{\d m_K}=\lambda_K\phi'_K(\beta_K)>0;\\
        \frac{\d R^*(\bm{m},\bm{s})}{\d s_K}=-\lambda_K<0,
    \end{dcases}\label{eqn:comp:stat}
\end{align}
 
where $\lambda_K>0,\beta_K>m_K$ are parameters pinned down in the proof of \Cref{thm:moment}.
\end{corollary}
\begin{proof}
    See derivation in \Cref{sec:ift}.
\end{proof}

% \yk{Do we need these para?  I am a little confused on the second para.}
% \wz{[\Cref{cor:comp:statics}] is w.r.t. m\&s. But our main model has a general $\S$. The second paragraph talks about comparative statics w.r.t. general $\S$}

Without loss, we can normalize $\phi_K$ so that it is centered at $m_K$ (i.e. $\phi'_K(m_K)=0$), then $\phi_K''>0$ implies $\phi_K'(\beta_K)>0$. We obtain a complete comparative statics result when $\S$ is a singleton: The optimal revenue guarantee decreases strictly in the dispersion of the valuation distribution and increases strictly in the mean of the valuation. 

Moreover, for general $\S$, \Cref{eqn:comp:stat} quantifies the tradeoff between the revenue guarantee gain from improving the ``precision of estimation'' for different summary statistics. Consider the special case when $\S$ is a product set---the seller knows a confidence interval for each mean and moment estimation---then, \Cref{eqn:comp:stat} implies that $\phi'_K(\beta_K)$ is exactly the marginal rate of substitution between shrinking the confidence interval of the mean versus the dispersion moment of a product group.\footnote{More precisely, lowering the upper bound of dispersion $s_K$ by a unit is revenue-equivalent to an increase in the lower bound of the mean $m_K$ by $\phi'_K(\beta_K)$ units.}

\section{Necessity of $\K$-bundled sales}\label{sec:necessity}
    \Cref{thm:moment} establishes the robust optimality of a $\K$-bundled sales mechanism.  However, it leaves open the possibility that another mechanism may attain that same revenue guarantee.  We next show that this is not the case and therefore the main qualitative feature of $\K$-bundled sales is {\it essential} for achieving that optimal revenue guarantee.
    This analysis also reveals what motivates the seller to choose the predicted mechanism.  

% In \Cref{sec:moment}, we established the robust optimality of a $\K$-bundled sales mechanism.  We explained its robustness to  the mechanism being immune to redistribution of value within  each category $K\in \K$ and to the perturbation of correlation across bundles. 
% These two features seem to uniquely pin down the $\K$-bundling structure. In this section, we prove that indeed this is the case and illustrate the fragility of alternative mechanisms.

\begin{theorem}  \label{thm:necessity}  Fix an ambiguity set $\F$ defined relative to the partition $\K$. 
\begin{enumerate}
    \item Suppose there is $K\in \K$ with $|K|\ge 2$.  Then, for any nonempty sets $J, J'\subset K$ with $J\cap J'=\emptyset$, no mechanism in $\M_{\K'}$ is robustly optimal if $\K'$ separates $J$ and $J'$;  i.e., $J,J'\in \K'$. 
    \item Let $\bm{\alpha},\bm{\beta}$ be the parameters derived in \Cref{thm:moment}. Suppose there are $K,K'\in \K$ such that $\beta_K/\alpha_K\ne \beta_{K'}/\alpha_{K'}$ (as defined in \cref{eq:mi':1} and \cref{eq:ki:1}).  Then, no mechanism  in $\M_{\K'}$ is robustly optimal if $\K'$ bundles $K$ and $K'$ together; i.e,   $K\cup K'\in\K'$. 
\end{enumerate}
\end{theorem}
While we relegate the formal proof of \Cref{thm:necessity} to \Cref{proof:sub}, we illustrate the intuition of the proof via \Cref{eg:1}, where $\K=\{\{1\},\{2,3\}\}$. In this case, \cref{thm:necessity} rules out two alternative types of mechanisms: the fully separated mechanism (corresponding to $\K'=\{\{1\},\{2\},\{3\}\}$) and the pure bundling mechanism (corresponding to $\K''=\{\{1,2,3\}\}$). Recall that, given $F^*$, any mechanism that allocates item $i$ when $v_i=\beta_i$ with probability 1 is optimal. Therefore, the two types of candidate mechanisms are optimal when the prices are chosen not too high. However, they are \textit{not} robustly optimal, as each of them is susceptible to a type of deviation by nature that lowers the revenue guarantee.  These deviations reveal the potential distribution $F$ that motivates the seller to choose the ``correct'' $\K$-bundled sales.

 \paragraph{Why is full separation not robustly optimal?} 
 Suppose instead of the $\{\{1\}, \{2,3\}\}$-bundling, the seller sells all three goods separately,  in particular, separating items 2 and 3. The separate sale is vulnerable to the following deviation by nature.
 Consider a distribution $\tilde F$, which is the same as $F^*$, except that a small mass $\varepsilon$ is transferred from  $(\beta_1,\frac{1}{2}\beta_{23},\frac{1}{2}\beta_{23})$ (the point mass at the top) to $(\beta_1,\beta_{23},0)$ and $(\beta_1,0,\beta_{23})$, each with respective masses of $\frac{1}{2}\varepsilon$. See \Cref{fig:F3-1}.

\begin{figure}[htbp]
	\centering
	\includegraphics[height=7cm]{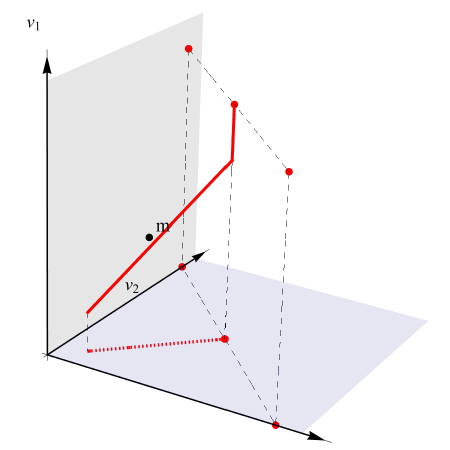}
	\caption{The support of distribution $\tilde{F}$}
	\label{fig:F3-1}
\end{figure}

This change keeps all constraints satisfied and does not alter the revenue of $M^*$ since the distributions of $v_1$ and $v_2+v_2$ remain the same. Yet the change has increased the dispersion of each individual item value in a way that makes separate sales less profitable. 
% To see this, suppose the seller sells all three items separately, in particular, separating items 2 and 3. 
For  $\varepsilon$ sufficiently small, the seller will never wish to charge prices $0$ or $\beta_{23}$ for either item 2 or item 3. For any other price in the support, the seller loses revenue $p\cdot \frac{1}{2}\varepsilon$, when compared with bundling items 2 and 3. Consequently, facing   distribution $\tilde F$, the seller earns strictly below  $\alpha_1+\alpha_{23}$ by selling the three items separately.  In essence, the fear of this ``negatively-correlated'' counterfactual distribution  motivates the seller to bundle goods 2 and 3.

\paragraph{Why is pure bundling not robustly optimal?}  Suppose now the seller bundles all three items.   As observed earlier, given the same $F^*$, the grand bundle yields the same optimal revenue when it is sold at price $p$ within $\big[\alpha_1+\alpha_{23},\beta_{23}\left(1+\frac{\alpha_1}{\alpha_{23}}\right)\big]$; see \Cref{fig:prop1:profit}.  However, the same figure hints at why selling the grand bundle is not robustly optimal.  Suppose the seller charges an even higher price   $p>\beta_{23}\left(1+\frac{\alpha_1}{\alpha_{23}}\right)$ for the bundle. Then, the revenue would be strictly lower!  This is because bundling entails inefficient screening at that price, specifically in the vertical segment of the support depicted in \Cref{fig:prop1:dist}: the purchases of all goods are now tied so that  the buyer will refuse to buy the bundle {\it even when} he has the highest value $\beta_{23}$ for the bundle $\{2,3\}$, if his value of good 1 is less than $p-\beta_{23}$, resulting in the seller  not being able to sell goods 2 and 3 in that case.  This is clearly inefficient and this inefficiency never occurs under separate sales of $\{1\}$ and $\{2,3\}$, since the seller would never charge more than $\beta_{23}$ for the bundle $\{2,3\}$.

\begin{figure}[htbp]
	\begin{minipage}[tb]{0.44 \textwidth}
		\centering
		\includegraphics[height=5cm]{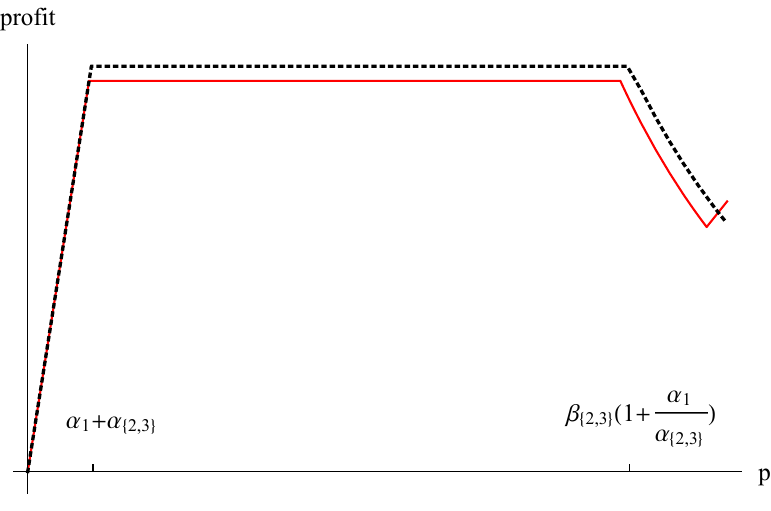}
		\caption{\small Profit from pure bundling mechanism under $F^*$ (dashed) and $\tilde{F}$ (solid)}
		\label{fig:prop1:profit}
	\end{minipage}
	\begin{minipage}[tb]{0.55\textwidth}
		\centering
		\includegraphics[height=5.2cm]{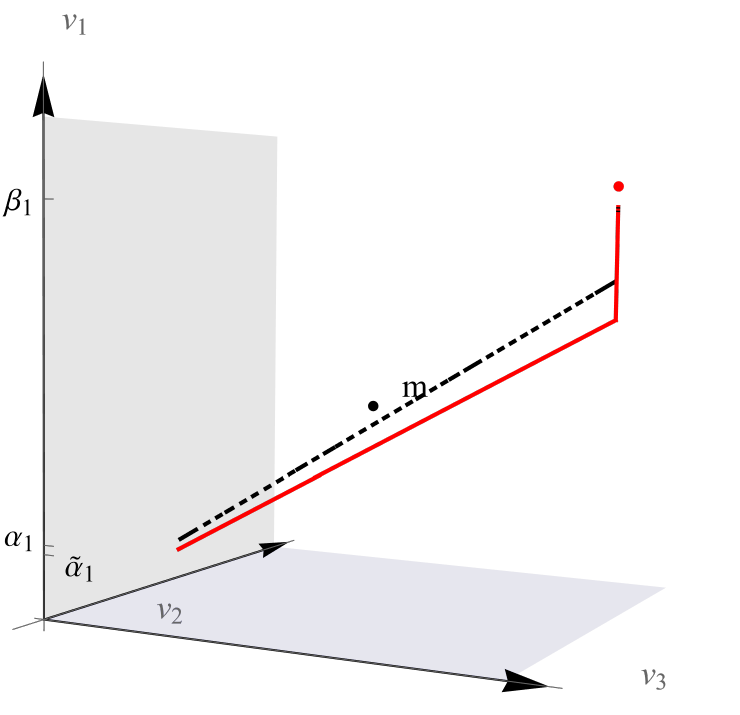}
		\caption{\small The support of   distribution $F^*$ (dashed) and $\tilde{F}$ (solid)}
		\label{fig:prop1:dist} 
	\end{minipage}
\end{figure}

% if his value of good 2 is low, which is inefficient from the seller's perspective in terms of surplus extraction. By contrast, separate sales avoid such an inefficient screening by unbundling the sale of two goods.\footnote{For instance, the buyer with the highest value $\beta_1$ of good 1 is always induced to buy that good.}

Nature can exploit this ``weakness'' of the grand bundling by shifting mass toward that vertical segment.  Consider a new distribution $\tilde{F}$ supported on the red curve in \Cref{fig:prop1:dist}. Compared with $F^*$, this new distribution lowers the infimum of $v_1$ from $\alpha_1$ to $\tilde{\alpha}_1$, thus lowering the value $V_1(s)$ of good 1 on the interior segment of the support. This reduces the revenue the seller can collect by charging a low bundle price $p$.  Of course, nature cannot  lower the value of good 1 uniformly across the board, because this will violate the mean condition.  To satisfy the latter, $\tilde{F}$ must therefore put larger mass at its supremum value ${\beta}_1$ of good 1.   %which in turn allows $\tilde{\beta}_2$  to be slightly lower than $\beta_2$.   
The seller cannot take advantage of this increased mass at ${\beta}_1$ under pure bundling since the profit at  $p$ in the neighborhood of $ \beta_1+\beta_{23}$ was strictly lower than $\alpha_1+\alpha_{23}$, as can be seen in \Cref{fig:prop1:profit}.\footnote{The robust optimality of $M^*$ means that  the seller would receive at least  $\alpha_1+\alpha_{23}$  from $M^*$  given  $\tilde{F}$.}  Hence, the new distribution keeps the seller's revenue strictly below $\alpha_1+\alpha_{23}$ no matter the price of the bundle.  Intuitively, the distribution $\tilde F$ exacerbates the ex ante asymmetry across the two bundles, and the {\it fear} of such an asymmetric distribution motivates the seller to choose   separate sales mechanism.

Analyzing the counterfactual suboptimal mechanisms teaches us two important lessons about the robustly optimal mechanism. First, it is well known that negatively-correlated item values make bundling desirable in the standard Bayesian context (see \cite{adams1976commodity}).  In light of this, one may find  it surprising that the item values under  distribution $F^*$  are instead  positively correlated.    \Cref{thm:necessity}-1 clarifies this issue:  it is a \emph{possible} negative correlation ``off the path'' that motivates the seller to use bundling in the current environment.\footnote{When we state ``off-the-path'', we are invoking the definition of $R^*$: the seller acts first, knowing nature's response.} Second, it is also well-known that perfectly correlated (comonotonic) item values are the worst-case distribution that justifies full separation under correlational ambiguity (see \cite{carroll2017robustness}). \Cref{thm:necessity}-2 further explains what deters bundling: nature may ``deviate'' to distributions that are still comonotonic but strongly asymmetric, which requires the seller to screen different dimensions asymmetrically to attain the maximal revenue.

\section{Informational Robustness} \label{sec:info}

 A rather surprising application of our analysis is informational ambiguity, where the source of ambiguity for the seller is not the prior on the buyer's valuations but rather the information the latter has about the valuations.  A growing number of recent papers study  mechanisms that are robust with respect to such ambiguity; see, for example, \cite{du2018commonvalue, brooksdu2019,brooks2021maxmin, roesler2017buyer,ravid2019learning,bbm2019revenueguarantee}.

To fix the idea, suppose the seller has prior distribution $G\in \Delta(\mathbb{R}_+^n)$ on the valuations of the goods, but she has ambiguity on the information the buyer himself has about the valuations.  By \cite{Blackwell1951}, a possible signal the buyer may have is characterized by a mean-preserving contraction of $G$.  Hence, one can describe the seller's ambiguity  by a set of convex moment constraints:  
\begin{align}  \label{eq:F-info}
\F_G:= \left\{ F\in \Delta(\R_+^n)\big| \, \mathbb{E}_F[\phi(\bm{v})]\le  \mathbb{E}_{G}[\phi(\bm{v})], \forall \phi \text{ convex} \right\}.
	\end{align}
 A robustly optimal mechanism given the ambiguity set $\F_G$ in \Cref{eq:F-info} is then called \textbf{informationally robust}. 
% Conceptually, one can identify the convex function $\phi$ that is most ``binding'' and apply \Cref{thm:moment} 
% to find the robustly optimal mechanism.    
Formulated in this manner, our analysis can be readily applied to find an informationally robust mechanism.  %In particular, we can generalize \cite{deb-roesler2021}'s Theorem 1 that pure bundling is informationally robust when the prior distribution $G$ is exchangeable. 
Consider the following assumption on the prior distribution $G$.   
\begin{assumption} [Stochastic Comonotonicity] \label{as:comon} There exists $(\xi_1, ..., \xi_n)\in \mathbb{R}_+^n$ with $\sum_i \xi_i=1$ such that for each $i=1,..., n$
$$\E\Big [ v_i\Big |\tsum_{j=1}^nv_j \Big] =\xi_i \left(\tsum_{j=1}^n v_j\right).$$ 
\end{assumption}

Stochastic comonotonicity means that the expected value of each item conditional on the total value of all items is simply a fixed fraction of the latter.  Geometrically, the conditional mean of each item value forms a linear ray, as depicted in \Cref{fig:info}. Effectively, the condition 
means that the item values are a garbling---a mean-preserving spread---of such a ray, as illustrated in \Cref{fig:info}.

\begin{figure}[htbp]
    \centering
    \includegraphics[width=0.4\linewidth]{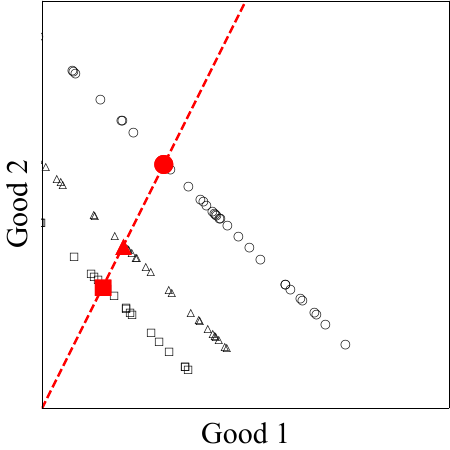}
    \caption{\small The red circle (triangle/square) is the conditional expectation of the values, conditional on the total value being the same as the red circle, illustrated by the black circles (triangles/squares). In this example, such conditional expectations are aligned on the ray $v_2=2v_1$.}
    \label{fig:info}
\end{figure}

\begin{theorem}  If the seller's  prior distribution $G$ is stochastically comonotonic, then   pure bundling is informationally robust.
	\label{thm:info}
\end{theorem}

To prove \Cref{thm:info},  we invoke \Cref{thm:moment} with $\K=\{N\}$.  The application is not trivial, however, since the ambiguity set $\F_G$  requires a collection of convex moment conditions instead of a single convex moment condition as required by the ambiguity set in \Cref{eq:ambiguity-partial}. Consequently, in order to apply \Cref{thm:moment}, we identify a single ``binding'' convex moment function of the form $\phi_N$---namely, one that applies to the total sum of values---out of all convex moment conditions in $\F_G$.  Let $\F$ be the ambiguity set that results from imposing only that binding condition.   Then, $\F$ conforms to the form assumed in \Cref{eq:ambiguity-partial} with $\K=\{N\}$, so \Cref{thm:moment} can be  applied against the ambiguity set $\F$.  This means that if the seller were to face $\F$ as the ambiguity set, then pure bundling is robustly optimal against the worst-case distribution $F^*$, which has comonotonic support. 

Note, however, that $\F$ is not the true ambiguity set; instead, the seller's ambiguity set is   $\F_G$, a subset of $\F$.  Here is where the stochastic comonotonicity of $G$ is required.  If $G$ is stochastically comonotonic, then $G$ is a mean preserving spread of $F^*$ (see    \Cref{fig:info} for an illustration), so $\F_G$ does indeed contain $F^*$.   This means that the pure bundling that forms a saddle point along with $F^*$ given $\F$ also forms a saddle point given $\F_G$ (by \Cref{cor:beyond}), which proves that it is max-min optimal against $\F_G$.

Stochastic comonotonicity is fairly general. In particular, it accommodates \emph{exchangeable prior} as a special case.  To see this, suppose $G$ is exchangeable; namely, {for all} permutations $(i_1,\dots,i_n)$ of $(1,\dots,n)$, $G(v_1,\dots,v_n)=$ $G(v_{i_1},\dots,v_{i_n})$.   Then, for each $i$,
$$\E\Big [ v_i\Big |\tsum_{j=1}^nv_j \Big]=\frac{1}{n}\tsum_{k}\E\left[ v_k\Big |\tsum_{j} v_j \right]=\frac{1}{n}\E\left[ \tsum_{k}v_k\Big |\tsum_{j} v_j \right]=\frac{1}{n} \tsum_{j=1}^n v_j,$$ 
so $G$ is stochastically comonotonic with $\xi_i=1/n$ for all $i$. Therefore, \Cref{thm:info} nests Theorem 1 of \cite{deb-roesler2021}, which proves the same result under exchangeable prior. This generalization is relevant both from conceptual and practical perspectives.  It is analytically important since it speaks to the essential feature of the prior that makes pure bundling informationally robust.  The reader of \cite{deb-roesler2021} may conclude that the exchangeability of the prior, with all the restrictions it involves, may be crucial for the result. In particular, one may wonder if symmetry is an important driver of what makes pure bundling robustly optimal.\footnote{\label{fn:complementary} \cite{deb-roesler2021} study two extensions that permit asymmetry across items. They show that (i) pure bundling is still robustly optimal when the value of a proper subset  $B\subsetneq N$ of items can be lower than $\sum_{i\in B}v_i$, and (ii) bundling a proper subset \(B\subsetneq N\) is robustly optimal when the values are ``exchangeable'' only within $B$. Nevertheless, both extensions rely on the symmetry of all items within the considered bundle.  Our \Cref{thm:info} can be easily extended to obtain (i) analogously to \Cref{cor:complementarity} and to obtain (ii) as is explained in the subsequent paragraphs.}
 Our analysis shows that the symmetry implied by an exchangeable prior is not crucial for pure bundling to be robustly optimal.  
 As is clear, stochastic comonotonicity allows for arbitrary asymmetry in terms of the mean values.  %\yk{Meanwhile, \Cref{thm:necessity}-2  highlights the importance of the comonotonically-supported distribution  for explaining pure bundling as a robustly optimal mechanism}  
 In this sense, our theorem uncovers the fundamental property of the prior that makes pure bundling informationally robust.

 Second, the generality gained from stochastic comonotonicity is not just significant,  it is also practically relevant.  
Stochastic comonotonicity is consistent with many simple models or heuristics used in various settings. For instance, consider an investment bank's problem in pricing the assets. There are $n$ assets, all belonging to a sector $N$.  The capital asset pricing model (CAPM) implies a prior distribution of the asset values consistent with stochastic comonotonicity.  To be concrete, suppose the bank subscribes to the model that assesses the return of each of  $n$ assets as:  
\begin{align*}
    r_i=\beta_i\cdot r_m+e_i,
\end{align*}
where the random variable $r_m$ is the ``market return'' of the sector, a constant $\beta_i$ is the ``beta'' of the asset and $e_i$ is the idiosyncratic risk satisfying $\sum e_i=0$ and $\E[e_i|r_m]=0$.  Since $\E[r_i|\sum_j r_j]=\beta_i \sum_j r_j$, such a model satisfies stochastic comonotonicity. If the investment bank is concerned with informational ambiguity, the max-min optimal policy would  be to sell its bundle as an asset.  

Finally, while \Cref{thm:info} only covers the case of $\K=\{N\}$, with pure bundling as the optimal solution, it is not difficult to generalize the theorem to obtain partial bundling as an informationally robust policy by expanding the nature of ambiguity. For any arbitrary partition $\K$, suppose that the seller only knows the marginal distribution $G^K$ for each (partial) bundle $K$ of items, without knowing anything at all about the correlation of the values across different $K$'s within $\K$. Further, the seller does not know the information the buyer may have about the item values.  In this case, the seller faces both distributional and informational ambiguity.   Applying \Cref{thm:moment}, one can show that the seller facing such ambiguity will find  $\K$-bundled sales as a robustly optimal strategy. See \cite{che-zhong2022} for details.  Cast in the investment banker example, if her prior is given by the sector-specific CAPM model, then an informationally robust strategy calls for bundling all assets in each sector $K$ and selling the alternative bundles separately. In other words, the optimal menu of portfolios contains the market return (i.e. a market \emph{index}) for each sector.

\section{General Distributional Ambiguity}\label{sec:extensions}

We have so far focused on ambiguity sets characterized by moment conditions. In this section, we go beyond moment conditions and identify a general structure of ambiguity sets that would give rise to the robust optimality of $\K$-bundled sales. Special cases will identify the conditions that justify separate sales and pure bundling. 

To state the general condition on the ambiguity set, we first define an operator  $\Upsilon_{\K}: F\mapsto  \Delta(\R_+)^{|\K|}$:
$$\Upsilon_{\K}(F):=\left\{( F_K)_{K\in\K}\in \Delta(\R_+)^{|\K|}: \forall K\in \K, \forall z\in \R_+, F_K(z):=\P_{F} \{\mbox{$\sum_{j\in K}  v_j\le z$}\}\right\},$$
where $\P_{F}\{\cdot\}:=\E_{F}[\pmb{1}_{\{\cdot\}}]$. In words, $\Upsilon_{\K}(F)$ calculates the marginal distribution of the total value of each bundle $K\in \K$, given the initial distribution $F$. $\Upsilon_{\K}(F)$ is called the  {\bf $\K$-marginals} of $F$.  

\begin{definition} Fix any arbitrary partition $\K$ of $N$.  An ambiguity set $\F\subset \Delta(\R^n_+)$ exhibits {\bf $\K$-Knightian ambiguity} if
	$\F= \Upsilon_{\K}^{-1}\circ \Upsilon_{\K}(\F).$
\end{definition}

The notion of $\K$-Knightian ambiguity assumes two types of ambiguity. First, the seller has arbitrary knowledge about the $\K$-marginals; thus, all $\K$-marginals $(F_K)$ in $\Upsilon_{\K}(\F)$ are considered possible. Second, for each tuple of $\K$-marginals that the seller considers possible, she faces full ambiguity about the joint distribution; thus, all joint distributions in $\Upsilon_{\K}^{-1}((F_K))$ are considered possible.  In particular, this means she faces  ambiguity on a) the correlation of total values of product groups $K$'s across those in $\K$ and b) the distribution of values across items within each product group $K\in \K$.

%Some special cases may be instructive to consider. Take the finest partition $\K=\{\{1\},...,\{n\} \}$.  In this case, $\K$-Knightian ambiguity means that regardless of the marginal distribution of item values the seller considers possible from her ambiguity set $\F$ she cannot use the structure of $\F$ to rule out some correlation structures across the item values. 
%Again, she need not know very little about the marginal distribution of item values.  An example is the case of \yk{section reference on separate means} where the seller only knows  the moments of item values. 
%Likewise, for the coarsest partition $\K=\{\{1,..., n\}\}$, $\K$-Knightian ambiguity means that any information the seller has about the total value of all goods cannot be leveraged to mitigate  ambiguity about the associated item values.  

A special case of   $\K$-Knightian ambiguity is the case studied in 
\Cref{sec:model} where the seller knows only moments of the $\K$-marginals. But there are many other examples. For instance, $\K$-marginals may be constrained such that  $(F_K)_{K\in \K}\in \mathcal{G}\subset \Delta(\R_+)^{|\K|}$, for some arbitrary  set $\mathcal{G}$.  We list specific examples of $\K$-Knightian ambiguity:   

\begin{itemize}
    \item For each $K$, the ambiguity set may include every $F_K$ within a distance, say $\delta_K>0$, from some reference marginal distribution $F_K^0$ the seller finds plausible.\footnote{The metric could be sup norm or Levy-Prokhorov, among others.} \cite{bergemannschlag2008pricing} formulated ambiguity in this sense.
    \item For each $K$, the ambiguity set  may require $\underline F_K\le_{SO} F_K \le_{SO}  \overline F_K$ for some benchmark distributions $\underline F_K, \overline F_K$ and some arbitrary stochastic order $\le_{SO}$ that is closed under convex combinations.  Examples of such stochastic orders are  First-Order Stochastic Order, Second-Order Stochastic Dominance, Lehmann, Supermodularity, or combinations thereof.\footnote{Recall, however, from \Cref{sec:info} that the Second-Order Stochastic Dominance Order, or equivalently the Mean Preserving Spread Order, can be handled by dispersion moment conditions involving particular (piece-wise linear) convex moment functions.}
\end{itemize}    
   
In addition to the knowledge specified by $\F$, we allow the seller to have arbitrary knowledge about the means of item values.  Specifically, consider a set 
$$\widehat{\F}:=\{F\in\Delta(\R_+^n):  (\mu_1(F),..., \mu_n(F))\in \Omega\},$$
where $\Omega$ is an arbitrary nonempty subset of $\R_{++}^n$. We then assume that the seller's ambiguity set is given by $\F\bigcap \widehat{\F}$. Clearly, when $\Omega=\R_+^n$, the constraint specified by $\widehat{\F}$ has no bite at all. 

The main result requires some technical assumptions.  We say a  set $\F'$ of distributions is {\bf regular} if \emph{$\F'$ is nonempty, convex, closed under weak topology, tight, and has bounded expectation.}\footnote{A set of measures on $\mathbb{R}_+^n$ is tight if for any $\epsilon>0$ there is a compact subset $S\subset  \mathbb{R}_+^n$ whose measure is at least $1-\epsilon$.  All other notions are standard.}  Our main theorem then follows:

\begin{theorem}	
	\label{thm:partial:bundling} Fix any partition $\K$ of $N$. Suppose the seller faces a regular ambiguity set $\F\bigcap \widehat{\F}$, where $\F$ exhibits $\K$-Knightian ambiguity. Then, a $\K$-bundled sales mechanism is robustly optimal in the sense that
	\begin{align*}
	\sup_{M\in\M}\inf_{F\in\F\bigcap \widehat{\F}}R(M,F)=\sup_{M\in\M_{\K}}\inf_{F\in\F\bigcap \widehat{\F}}R(M,F).
	\end{align*}
\end{theorem}

\begin{proof}  See \Cref{sec:proof-saddle}.
\end{proof}

 $\K$-Knightian ambiguity crystallizes the insight that gives rise to separation and bundling in the earlier section. Specifically, the concept captures the ambiguity about how the values of alternative bundles in $\K$ are correlated and the ambiguity about how a given value of a bundle $K\in\K$ is distributed across items within $K$. The former gives rise to the separation of sales across alternative bundles in $\K$ whereas the latter ambiguity gives rise to the bundled sales of items within each $K$. 

As special cases, the theorem provides conditions for the robust optimality of two canonical sales mechanisms:\footnote{ \Cref{cor:canonical} part 1 formalizes the conjecture in p. 481 of \cite{carroll2017robustness}: when $\G=\Upsilon_{\K}(\F)$ is ``well-behaved enough to contain a single worst marginal distribution,'' then separate sale is robustly optimal.  Indeed, even when the seller's ambiguity set contains a non-singleton set of marginal distributions, if it admits a unique saddle point, then the optimal mechanism in the saddle point must be robustly optimal against the exact marginal distributions associated with that saddle point, so the mechanism must be separating, following Theorem 1 of \cite{carroll2017robustness}.  However, it is not easy to guarantee existence or uniqueness of a saddle point.  For instance, our regularity condition does not necessarily lead to the existence of a saddle point or its uniqueness.  We therefore did not follow  \cite{carroll2017robustness}'s conjectured recipe of establishing (unique) saddle points for each item.  Our regularity condition, although not sufficient for existence of  a saddle point, guarantees the robust optimality of $\K$-bundled sales mechanisms, following a minimax argument.}

\begin{corollary}\label{cor:canonical}  The seller's ambiguity set $\F$  exhibits $\K$-Knightian ambiguity.
\begin{enumerate}
    \item If $\K$ is the finest partition of $N$ and $\F$ is regular, then separate sales are robustly optimal.
    \item If $\K$ is the coarsest partition of $N$ and $\F\bigcap \widehat{\F}$ is regular, then pure bundling is robustly optimal.
\end{enumerate}
\end{corollary}

\Cref{thm:partial:bundling}  %, as well as \Cref{cor:canonical}, 
identifies $\K$-Knightian ambiguity as a fundamental general condition for $\K$-bundled sales to be robustly optimal.  To the best of our knowledge, this condition provides for the most general characterization of the extent to which items should be bundled or separated. Since $\K$-Knightian ambiguity holds under the moment restrictions considered in \Cref{sec:moment}, this condition can be seen as responsible for the robust optimality found in that section.\footnote{\label{fn:lm} The moment conditions required by $\F$ in \Cref{sec:moment} clearly satisfies $\K$-Knightian ambiguity.  We prove in \Cref{sec:add-lemma} that $\F$ considered in \Cref{sec:moment} is regular.} 
Nevertheless,  \Cref{thm:partial:bundling} does not make that section superfluous. Note that the current theorem does not identify the exact form of the optimal mechanism or the worst-case distribution, whereas the additional structure given by moment restrictions allowed us to identify them in \Cref{thm:moment}.  Not only is the exact identification of the mechanism and distribution important and useful of its own right, it enables us to go beyond $\K$-Knightian ambiguity, which is sufficient but not necessary for $\K$-bundled sales to be robustly optimal.  For instance, as noted by \Cref{cor:beyond} and \Cref{thm:info}, the exact solution of the joint distribution enables us to identify a robustly optimal mechanism---i.e.,  $\K$-bundled sales---even when the ambiguity set $\tilde \F$ fails $\K$-Knightian ambiguity.  Finally, the worst-case distribution found in \Cref{thm:moment} plays a crucial part in the proof of  \Cref{thm:partial:bundling}, which makes the former indispensable for obtaining the current generalization.

\section{Concluding Remarks} \label{sec:conclusion}

The current paper has characterized   robustly optimal mechanisms for selling multiple goods for a monopolist faced with ambiguity on the buyer's private valuations of the goods. The nature of the robustly optimal mechanism depends on the type of ambiguity facing the seller.  We have identified moment conditions as well as general distributional conditions leading to the robust optimality of a  $\K$-bundled sales mechanism, which includes the commonly used sales mechanisms of separate sales and pure bundling as two special cases.  The distributional condition that we identify, namely, $\K$-Knightian ambiguity,   is  the most general kind known to date that rationalizes these sales mechanisms. More importantly, the concept captures the clear economic insights that give rise to separation and bundling of items in a (robustly) optimal sale. As argued in detail,  ambiguity about the correlation of values across items/bundles leads to separation of items/bundles, whereas  ambiguity about across-items value dispersion leads to the bundling of items in the sale.  In particular, the latter ambiguity features the threat of negatively-correlated item values as a reason for favoring a bundled sales, thus connecting with the classic insight provided by \cite{adams1976commodity}.

Carrying the theme of \cite{carroll2017robustness} to its fruition, the current paper thus provides a general robustness perspective on the rationale for alternative canonical sales mechanisms.  As such, it offers a complementary as well as an alternative perspective on the subject matter which has so far been approached almost exclusively from a Bayesian mechanism design perspective.

There are at least two avenues along which one could further extend the current paper.  First, our model, like all other papers on the subject matter, assumes a single buyer, and naturally, one might consider introducing multiple buyers into the model.  Two concurrent papers have made progress under such generality; however, there is still no general answer.\footnote{A few papers identify robustly optimal mechanisms in  {\it single-item} auctions.  \cite{brooks2021maxmin} finds a robustly optimal auction mechanism, when  robustness is required with respect to  value distributions with known means and common domain, buyers' high-order beliefs, and to equilibrium selection.
Considering a similar mean constraint but restricting attention to the private-value setting, \cite{che2022robustly} identifies a robustly optimal auction mechanism within a class of ``competitive'' mechanisms which encompass standard auctions. Similarly, \cite{bbm2019revenueguarantee} identifies an informationally-robust optimal auction mechanism in the class of symmetric and  standard auctions.} \cite{zhong2021auction} provide a limiting result that as the number of ex-ante identical buyers grows large, it is {\it asymptotically} robustly optimal to auction off via a second-price format the robustly optimal bundle identified in the current paper. \cite{brooks2021structure} provide a duality result that enables the calculation of the \emph{informationally} robust revenue guarantee via linear programming and show that it is without loss to consider an auction format with one-dimensional message space.

Second, while the current paper offers  a robustness-based rationale for  separate sales and pure bundling as well as more general $\K$-bundling, we do not offer a  rationale for so-called ``mixed-bundling,'' i.e., a menu of options for buying goods both separately and a bundle. Although the nature of  ambiguity that would justify such a mechanism remains unknown, we hope our current paper will offer  useful insights for   future inquiry into this topic.

\renewcommand{\baselinestretch}{1.15}
\bibliographystyle{apalike}
\bibliography{reference}
%\nocite{*}
\renewcommand{\baselinestretch}{1.25}
\addtolength{\jot}{-0.3em}
\setlength{\abovedisplayskip}{4pt}
\setlength{\belowdisplayskip}{3pt}

\appendix
\counterwithin{proposition}{section} 
\counterwithin{lemma}{section} 
\counterwithin{claim}{section} 
\counterwithin{remark}{section}
\counterwithin{theorem}{section} 

 \section{Appendix: Proofs}

\subsection{Proof of \Cref{thm:moment}} \label{Proof:moment}

\begin{proof}
We prove \Cref{thm:moment} by explicitly constructing a $\K-$bundled sales mechanism $M^*$ and a valuation distribution $F^*$ and verify that they constitute a saddle point. 
 
 \textbf{Construction of $F^*$.}    We first construct $F^*$, nature's choice of distribution. This involves two steps. We first fix an arbitrary pair $(\bm{m},\bm{s})\in \S$ and construct a distribution $F^{(\bm{m},\bm{s})}$ with means and dispersion characterized by $(\bm{m},\bm{s})$. We later describe how $({\bm{m}},\bm{s})$ is chosen. 
 To begin, let $m_K:=\sum_{j\in K}m_j$ for each $K\in \K$, and let $K(i)=\{K\in\K: i\in K\}$ denote the bundle containing item $i$.

Let $X$ be a random variable distributed from $[1, \infty)$ according to a cdf $H$:
 \begin{align*}
	H(x):=\text{Prob}[X\le x]=1-\frac{1}{x}.
\end{align*}
 Then, the value of item $i\in K$, for $K\in \K$, is given by: 
\begin{align*}  
	V_{i}(X):=\min\left\{ \alpha_{K} X,  \beta_{K}  \right\}\cdot\frac{m_i}{m_{K}}, 
\end{align*}
where, for each $K\in\K$, the parameters  $0<\alpha_K< m_K<\beta_K$ satisfy:
		\begin{align}
			\int_1^{\frac{\beta_K}{\alpha_K}}\frac{\alpha_K}{x}\d x &+\alpha_K=m_K; \label{eq:mi':1}\\
			\int_1^{\frac{\beta_K}{\alpha_K}}\frac{\phi_K(\alpha_K x)}{x^2} \d x &+\frac{\phi_K(\beta_K)\alpha_K}{\beta_K}=s_K. \label{eq:ki:1}
		\end{align}
		In short, the total value of each product group $K\in \K$ rises co-monotonically and linearly with the common random variable $X$ at rate $\alpha_{K}$; the value of each item $i$ is then determined in proportion to its mean $m_i$ relative to the total mean of the group value.  The parameters  $(\alpha_K,\beta_K)_{K\in \K}$ are in turn determined to satisfy the moment conditions with respect to the sum of means $m_K$  and dispersion $s_K$. \Cref{lem:alphabeta} guarantees that such a pair exists for any $(m_K,s_K)\gg (0,\phi_K(m_K))$.

By continuity of $ (\alpha_K)_{K\in\K}$ and compactness of $\S$, there exists $$(\bm{m},\bm{s})=\arg\min_{(\tilde{\bm{m}}, \tilde{\bm{s}})\in \S} \sum_{K\in\K}\alpha_K(\tilde m_K,\tilde s_K).$$  Setting $F^*:=F^{(\bm{m},\bm{s})}$  completes the construction of nature's choice of distribution.\par

\textbf{Construction of $M^*$.} Next, we define the candidate optimal mechanism $M^*$. In a nutshell, the seller sells each bundle $K$ separately at a random price distributed according to $G_K$. The corresponding direct mechanism is:
\begin{align*}
\begin{dcases}
	q^*_i(\bm{v})=G_{K(i)}\left(\tsum_{j\in K(i)} v_j\right),\\
	t^*(\bm{v})=\sum_{K\in \K}\int_{p\le \sum_{j\in K} v_j}pG_K(\d p).
\end{dcases}
\end{align*}
The cdf $G_K$ is defined via the density function:
\begin{align*}
	g_K(v):=\lambda_K\cdot \frac{\phi'_K(\beta_K)-\phi'_K\left( v \right)}{v}
\end{align*}
on $[\alpha_K, \beta_K]$ and zero elsewhere, where 
$\lambda_K:= 1/[\int_{\alpha_K}^{\beta_K}\frac{\phi'_K(\beta_K)-\phi'_K(x)}{x}\d x]  $ normalizes the density so that it integrates to one. As is illustrated in \Cref{eg:1}, $g_K$ is constructed so that the revenue from selling bundling $K$ is an affine transformation of the moment function $\phi_K$.

\paragraph{Verification of saddle point}

  We first compute the value $R(M^*,F^*)$. For any $\bm{v}$ in the support of $F^*$,
	\begin{align*}
		t^*(\bm{v})=& \sum_{K\in \K}\int^{\sum_{j\in K}v_j}_{\alpha_K} pg_K(\d p)\\
		=&\sum_{K\in\K} \lambda_K \left \{\phi'_K(\beta_K)\left( \mbox{$\sum_{j\in K}$} v_j-\alpha_K \right)-\phi_K\left( \mbox{$\sum_{j\in K}$} v_j \right)+\phi_K(\alpha_K)\right\}.
		\end{align*}
		Hence,
		{\medmuskip=0mu
\thinmuskip=0mu
\thickmuskip=0mu
		\begin{align} R(M^*,F^*)=&\int t^*(\bm{v})F^*(\d \bm{v})\cr
	=&\sum_{K\in \K} \lambda_K \left\{\phi'_K(\beta_K)\left( m_K-\alpha_K \right)+\phi_K(\alpha_K)-\int\left( \phi_K\left( \mbox{$\sum_{j\in K}$}v_j\right)\right)F^*(\d \bm{v})\right\}\cr
	=&\sum_{K\in \K} \lambda_K\left\{\phi'_K(\beta_K)\left( m_K-\alpha_K \right)+\phi_K(\alpha_K)-s_K\right\}\cr
	=&\sum_{K\in \K} \frac{\phi'_K(\beta_K)\alpha_K\log(\beta_K/\alpha_K)-\alpha_K\int_{\alpha_K}^{\beta_K}\frac{\phi'_K(x)}{x}\d x}{\int_{\alpha_K}^{\beta_K}\frac{\phi'_K(\beta_K)-\phi'_K(x)}{x}\d x}=\sum_{K\in \K} \alpha_K.\label{eq:saddle}
	\end{align}}
	The first three equalities are straightforward. The fourth equality follows from \Cref{eq:mi':1,eq:ki:1} and from recalling that $\lambda_K= 1/[\int_{\alpha_K}^{\beta_K}\frac{\phi'_K(\beta_K)-\phi'_K(x)}{x}\d x]$.

	Next, we show that $M^*\in\arg\max_{M\in \M} R(M,F^*)$.  Fix any $ M=(q,t)\in\M$. Since the support of $F^*$ is a parametric curve $\bm{V}(x)$, the mechanism $M$ can be represented equivalently via $(\psi(x),\tau(x)):=(q(\bm{V}(x)),t(\bm{V}(x)))$.  Since $M$ satisfies $(IC)$, it must satisfy the envelope condition:
	\begin{align*}
	\tau(x)=&\psi(x)\cdot \bm{V}(x)-\int_1^x \psi(z)\cdot \bm{V}'(z)\d z.
	\end{align*}
	Hence,
		\begin{align}
		 R(M,F^*) =&   \int \tau(x) H(\d x)\cr
		 \le&\sup_{\psi(\cdot)} \int \psi(x)\cdot\left( \bm{V}(x)-\bm{V}'(x)\frac{1-H(x)}{h(x)} \right)H(\d x)\cr
	=&\sup_{\psi}\sum_{i}\int_1^{\frac{\beta_{K(i)}}{\alpha_{K(i)}}}\psi_i(x)\cdot 0 H(\d x)+\int_{\frac{\beta_{K(i)}}{\alpha_{K(i)}}}^{\infty} \psi_i(x)\cdot \gamma_i\cdot \beta_{K(i)} H(\d x)\cr
	\le&\sum_{i}\gamma_i\cdot\beta_{K(i)}\cdot\frac{\alpha_{K(i)}}{\beta_{K(i)}}=\sum_{K\in\K} \alpha_K
	=R(M^*,F^*), \label{eq:saddle1}
	\end{align}
where $\gamma_i:= \frac{m_i}{\sum_{j\in N} m_j}$. The second inequality follows from $\psi_i\le1$. The third equality follows from $\sum_{i\in N}\gamma_i=1$. The last equality follows from \cref{eq:saddle}.
	
 Finally,  we show that $F^*\in\arg\min_{F\in\F} R(M^*,F)$.   To this end, observe 
	\begin{align*}
		t^*(\bm{v})\ge&\sum_{K\in\K} \lambda_K \left\{\phi'_K(\beta_K)\left( \mbox{$\sum_{j\in K}$} v_j-\alpha_K \right)-\phi_K\left( \mbox{$\sum_{j\in K}$}v_j\right)+\phi_K(\alpha_K)\right\}.
		\end{align*}
	To see why this inequality holds, observe first that $t^*(\bm{v})=RHS$ when $\sum_{j\in K}v_j\in[\alpha_K,\beta_K]$ (recall the very first displayed equation in the proof). Outside that region, $t^*(\bm{v})$ is constant in $\bm{v}$, while the RHS is strictly decreasing in $\sum_{j\in K}v_j$ when   $\sum_{j\in K}v_j>\beta_K$ and strictly increasing in $\sum_{j\in K}v_j$ when $\sum_{j\in K}v_j<\alpha_K$. It then follows that, for any $F\in \F$,
	{
		\thinmuskip=0mu
\medmuskip=0mu
\thickmuskip=0mu
\small
	\begin{align*}
		R(M^*,F)=& \int t^*(\bm{v})F(\d \bm{v})\\
		\ge&\int\sum_{K\in \K} \lambda_K \left\{\phi'_K(\beta_K)\left( \mbox{$\sum_{j\in K}$} v_j-\alpha_K \right)-\phi_K\left(\mbox{$\sum_{j\in K}$} v_j \right)+\phi_K(\alpha_K)\right\} F(\d \bm{v})\cr
		=&\sum_{K\in\K} \lambda_K \left\{\phi'_K(\beta_K)\left( \mbox{$\sum_{j\in K}$} \E_F[v_j]-\alpha_K \right)+\phi_K(\alpha_K)\right\}-\sum_{K\in\K}\lambda_K\int \phi_K\left( \mbox{$\sum_{j\in K}$}v_j\right) F(\d \bm{v})\cr
		=& \underbrace{\sum_{K\in \K}  \lambda_K \left\{\phi'_K(\beta_K)\left(m_K-\alpha_K \right)+\phi_K(\alpha_K)-s_K\right\}}_{A}\cr
		 &-\underbrace{\sum_{K\in\K} \lambda_K \left\{\phi'_K(\beta_K)(m_K-\mbox{$\sum_{j\in K}$}\E_F[v_j])+\int \left( \phi_K\left( \mbox{$\sum_{j\in K}$}v_j \right)-s_K \right)F(\d \bm{v})\right\}}_{B}.
	\end{align*}
}

Note that \cref{eq:mi':1,eq:ki:1}, together with  $\lambda_K= 1/[\int_{\alpha_K}^{\beta_K}\frac{\phi'_K(\beta_K)-\phi'_K(x)}{x}\d x]$, imply that $A=\sum_{K\in\K} \alpha_K=R(M^*,F^*)$.

The above inequalities imply that $R(M^*,F)\ge R(M^*,F^*)-B$. If $F$ has the same moments as $F^*$, then $B=0$, so we are done.  Hence, assume $F$ has different moments than $F^*$.  Suppose for the sake of contradiction that $R(M^*,F)<R(M^*,F^*)$.
Since $\F$ is a convex set, if we define $F^{\delta}=F^*+\delta(F-F^*)$, then $F^{\delta}\in \F$ for any $\delta\in[0,1]$.   Since $R(M^*,F)$ is linear in $F$, we must then have $\frac{\d R(M^*,F^{\delta})}{\d \delta}\big|_{\delta=0}<0$. In particular, this implies that $\frac{\d B}{\d \delta}\big|_{\delta=0}>0$.  However, one can show that 
	\begin{align}
 \frac{\d B}{\d \delta}\Big|_{\delta=0}=&-\frac{\d \sum_{K\in\K}\alpha_K}{\d \delta}\Big|_{\delta=0}.  \label{eqn:computation}
\end{align}
(see  \Cref{sec:ift} for the details).  Hence, $\frac{\d B}{\d \delta}\big|_{\delta=0}>0$ means that $F^{\delta}$ entails a smaller value of  $\sum \alpha_K$  relative to $F^*$,  and thus lower revenue, for sufficiently  small $\delta$. However, this contradicts the fact that $\sum\alpha_K$ is minimized at $(\bm{m},\bm{s})$. Therefore, we conclude that $R(M^*,F)\ge R(M^*,F^*)$.
\end{proof}

		\begin{lemma}\label{lem:alphabeta} For any $(m_K,s_K)\gg (0,\phi_K(m_K))$ for each $K\in\K$, there exists a unique pair $(\alpha_K, \beta_K)$ satisfying \Cref{eq:mi':1} and \Cref{eq:ki:1}. The mapping $(m_K,s_K)_K \mapsto (\alpha_K)_K$ is continuous.
\end{lemma}
\begin{proof}  	From \Cref{eq:mi':1}, we can solve for $\beta_K=\alpha_Ke^{\frac{m_K-\alpha_K}{\alpha_K}}$. Substituting this into \Cref{eq:ki:1}, its LHS becomes   a continuous function of $\alpha_K$.  It is strictly decreasing in $\alpha_K$
	for any $\alpha_K<\beta_K$: 
	{\medmuskip=0mu
\thinmuskip=0mu
\thickmuskip=0mu
\small
	\begin{align*}
		\frac{\d\mbox{LHS of \Cref{eq:ki:1}}}{\d\alpha_K}
		=&\left( \frac{\phi_K(\beta_K)}{(\beta_K/\alpha_K)^2}-\frac{\phi_K(\beta_K)}{(\beta_K/\alpha_K)^2} \right)\cdot\frac{\d (\beta_K/\alpha_K)}{\d \alpha_K} +\int_1^{\frac{\beta_K}{\alpha_K}}\frac{\phi_K'(\alpha_Kx)}{x}\d x+\frac{\phi_K'(\beta_K)\alpha_K}{\beta_K}\cdot\frac{\d \beta_K}{\d \alpha_K}\\
		=&\int_1^{\frac{\beta_K}{\alpha_K}}\frac{\phi_K'(\alpha_Kx)}{x}\d x-\phi'_K(\beta_K)\frac{m_K-\alpha_K}{\alpha_K}\\
		=&\phi_K'(\beta_K)\left( \int_1^{\frac{\beta_K}{\alpha_K}}\frac{\phi'_K(\alpha_Kx)}{x\phi'_K(\beta_K)}\d x-\frac{m_K-\alpha_K}{\alpha_K} \right)\\
		<&\phi'_K(\beta_K)\left( \log \beta_K-\log\alpha_K-\frac{m_K-\alpha_K}{\alpha_K} \right)\\
		=&0,
\end{align*}}
 	where the strict inequality follows from the convexity of $\phi$, and the last equality is from substituting $\beta_K=\alpha_Ke^{\frac{m_K-\alpha_K}{\alpha_K}}$.

 	Observe next that the LHS of \Cref{eq:ki:1} is strictly less than its RHS when $\alpha_K=m_K$. It is strictly greater than the RHS when $\alpha_K$ is sufficiently low.  To see this, note
	{\medmuskip=0mu
\thinmuskip=0mu
\thickmuskip=0mu
\small
	\begin{align*}
		\int_1^{\frac{\beta_K}{\alpha_K}} \frac{\phi_K(\alpha_K x)}{x^2} \d x\ge& \int_1^{\frac{\beta_K}{\alpha_K}} \left( \frac{\phi_K(0)}{x^2}+\frac{\phi'_K(0)(\alpha_Kx)}{x^2}+\frac{1}{2}\varepsilon\frac{(\alpha_Kx)^2}{x^2} \right)\d x\\
		=&\phi_K(0)\left( 1-\frac{\alpha_K}{\beta_K} \right)+\phi_K'(0)\alpha_K(\log \beta_K-\log\alpha_K)+\frac{1}{2}\varepsilon\alpha_K(\beta_K-\alpha_K)\\
		\ge&-\left|\phi_K(0)\right|+\phi_K'(0)(m_K-\alpha_K)+\frac{1}{2}\varepsilon \alpha_K^2\left( e^{\frac{m_K-\alpha_K}{\alpha_K}}-1 \right).
\end{align*}}
 	The last line tends to $\infty$ as $\alpha_K\to 0$.  
 	
 	Collecting the observations so far, we conclude that there exists a unique pair $(\alpha_K, \beta_K)$ satisfying \Cref{eq:mi':1} and \Cref{eq:ki:1}. \par

	It is easy to see that both sides of \Cref{eq:mi':1,eq:ki:1} are continuous in $(\alpha_K,\beta_K,m_K,s_K)$. Therefore, $\alpha_K$ and $\beta_K$ each as a correspondence of $(m_K,s_K)$ has a closed graph. Since we have shown that $\alpha_K(m_K,s_K)$ is a function, it is continuous.  \end{proof}

\subsubsection{Proof of \Cref{cor:complementarity}}
 \label{proof:complementarity}
Let $(M^*,F^*)$ be the saddle point pinned down in \Cref{thm:moment}. It is straightforward that $F^*$ remains optimal within $\F$ because $M^*$ remains incentive compatible as the values for all bundles within $\K$ remain unchanged. Then, the revenue from $M^*$ does not change with $u_K$ for any $F$. 

Now, we prove that $M^*$ remains optimal. Note that since buyer's value is no longer additive, the definition of an allocation should be generalized to $\psi:x\to [0,1]^{2^N}$, i.e. probability of allocating each bundle, subject to $\forall i, \sum_{K\ni i} \psi_{K}(x)\le 1$. Then, given $M=(\psi(x),\tau(x))$ the envelope condition implies:
\begin{align*}
    \tau(x)=&\sum_{K\in 2^N}u_K\cdot \psi_K(x)\cdot\sum_{i\in K}V_i(x)-\int_x^x\sum_{K\in 2^N}u_K\cdot \psi_K(z)\cdot\sum_{i\in K}V_i(z)\d z\\
    \implies R(M,F^*)\le& \sup_{\psi(\cdot)}\int \sum_{K\in 2^N} u_K\cdot\psi_K(x)\cdot \sum_{i\in K} \left(V_i(x)-V_i'(x)\frac{1-H(x)}{h(x)}\right)H(\d x)\\
    =&\sup_{\psi(\cdot)}\sum_i\int_{\frac{\beta_{K(i)}}{\alpha_{K(i)}}}^{\infty}\sum_{K\ni i}\psi_K(x)u_K\gamma_i\beta_{K(i)}H(\d x)\\
    \le&\sum_i \gamma_i \alpha_{K(i)}=R(M^*,F^*).
\end{align*}
The second inequality is from $\sum_{K\ni i} \psi_{K}(x)\le 1$ and $u_K\le 1$.

 \subsection{Proof of \Cref{thm:necessity}} \label{proof:sub}

\textbf{(Part 1)} Fix any nonempty $J, J'\subset K$ for some $K\in\K$ such that  $J\cap J'=\emptyset$. We show that it is never robustly optimal to separate $J$ and $J'$.  To this end, it suffices to find $\widetilde{F}\in\F$ such that $\sup\limits_{M\in \M_{\K'}}R(M,\widetilde{F})<R(M^*,F^*)$, for any partition $\K'$ such that $\left\{ J,J' \right\}\subset \K'$.  

We construct $\widetilde{F}$ as follows. Define CDF $H$ and $\left( \alpha_K,\beta_K \right)$ as in \Cref{thm:moment}.  Recall  $X\sim H$.  Define a new binomial random variable $Y$ whose value is zero  with probability $\frac{m_{J'}}{m_{J\cup J'}}$ and one with probability $\frac{m_{J}}{m_{J\cup J'}}$. Let $0<\varepsilon<\min_K \frac{\alpha_K}{\beta_K}$.  The distribution $\widetilde F$ is then defined by the item values:
\begin{align*}
	V_i(X,Y)=
	\begin{dcases}
		\min\left\{ \alpha_{K(i)} X,\beta_{K(i)} \right\}\cdot \frac{m_i}{m_{K(i)}} &\text{if } i\not\in J\cup J'\\
		\min\left\{ \alpha_{K(i)} X,\beta_{K(i)} \right\}\cdot \frac{m_i}{m_{K(i)}} &\text{if }i\in J\cup J' \text{ and }    X\le 1/{\varepsilon}\\
		\beta_{K(i)}\cdot \frac{m_i}{m_{K(i)}}\cdot\frac{m_{J\cup J'}}{m_{J}}\cdot Y  &\text{if }i\in J\text{ and } X>1/{\varepsilon}\\
		\beta_{K(i)}\cdot \frac{m_i}{m_{K(i)}}\cdot\frac{m_{J\cup J'}}{m_{J'}}\cdot (1-Y) &\text{if }i\in J' \text{ and }  X> 1/{\varepsilon},
	\end{dcases}
\end{align*}
where recall $K(i):=K\in \K$ such that $i\in K$.  In words, the values of items $i\not\in J\cup J'$ are distributed same as $F^*$.  The values of $j\in J\cup J'$ are also distributed same as $F^*$ {\it conditional on $X<1/{\varepsilon}$}, an event that occurs with probability $H(1/{\varepsilon})=1-{\varepsilon}$.  In the complementary event, the value of good $j\in J$ becomes either $\beta_{K(i)}\cdot \frac{m_i}{m_{K(i)}}\cdot\frac{m_{J\cup J'}}{m_{J'}}$ or zero. Effectively, mass ${\varepsilon}$ of value $\beta_{K(i)}\cdot \frac{m_i}{m_{K(i)}}$ is split into a higher value and zero so that the expected value remains the same. Note that $j\in J'$ is split in the same fashion but in a way perfectly negatively correlated as the value of item $i\in J$.  The negative correlation means that the dispersion of values of group $K$ remains the same; recall both $J$ and $J'$ are in $K$.  Hence, all moment conditions of \cref{eq:ambiguity-partial} continue to be satisfied (since $F^*$ satisfies them). Therefore, $\widetilde{F}\in \F$.

Since the mechanism $M$ separates $J$ and $J'$, one can write: $$M=\left( q^{-J\cup J'} (v_{i,i\not\in J\cup J'}), t^{-J\cup J'}(v_{i,i\not\in J\cup J'}), q^J(\sum_{i\in J}v_i),t^J(\sum_{i\in J}v_i),q^{J'}(\sum_{i\in J'}v_i),t^{J'}(\sum_{i\in J'}v_i) \right).$$

In words, groups $J$ and $J'$ are each bundled separately, and the mechanism can be arbitrarily defined on all other items. The IC and IR conditions imply that $(q^{-J\cup J'}, t^{-J\cup J'})$, $(q^J,t^J)$ and $(q^{J'},t^{J'})$ should each satisfy IC and IR.    For $(q^{-J\cup J'}, t^{-J\cup J'})$, since the random vector $\bm{V}$ is effectively uni-dimensional for $i\not\in J\cup J'$, the envelope condition implies:
\begin{align*}
	\int t^{-J\cup J'}(\bm{v})F(\d \bm{v})\le& \sup_{\psi}\sum_{i\not\in J\cup J'}\int_0^{\frac{\beta_{K(i)}}{\alpha_{K(i)}}}\psi_i(x)\cdot 0 H(\d x)+\int_{\frac{\beta_{K(i)}}{\alpha_{K(i)}}}^{\infty}\psi_i(x)\phi_i \beta_{K(i)}H(\d x)\\
	\le&\sum_{i\not\in J\cup J'} \alpha_{K(i)}\frac{m_i}{m_{K(i)}}.
\end{align*} 
The sub-mechanisms
$(q^J,t^J)$ and $(q^{J'},t^{J'})$ sells bundles $J$ and $J'$ separately. For $\varepsilon>0$ sufficiently small, it is suboptimal to charge price $0$ for each bundle. However, charging any other price that leads to positive probability of sales generates revenue of 
\begin{align*}
	\alpha_{K}\frac{m_{J}}{m_K}\left( 1-\frac{m_{J'}}{m_{J\cup J'}}\varepsilon \right)
\end{align*}
from the sale of bundle $J$ (where $J\subset K$). Likewise, the sale of bundle $J'$ results in the revenue at most of 
\begin{align*}
	\alpha_{K}\frac{m_{J'}}{m_K}\left( 1-\frac{m_J}{m_{J\cup J'}}\varepsilon \right).
\end{align*}
Therefore,
\begin{align*}
	R(M,\widetilde{F})<& \sum_{i\not\in J\cup J'} \alpha_{K(i)}\frac{m_i}{m_{K(i)}} + \alpha_K \frac{m_J}{m_K}+\alpha_K\frac{m_{J'}}{m_K}\\
	=&R(M^*,F^*).
\end{align*}

 \textbf{(Part 2)}  For each $K\in \K$, let $\ell_K:=\frac{\beta_K}{\alpha_K}$. Suppose there are $K,K'\in\K$ such that  $\ell_{K}\neq \ell_{K'}$.  We will show that it is never robustly optimal for the seller to bundle goods in $K\cup K'$.  It suffices to find $\widetilde{F}\in\F$ such that $\sup_{M\in \M_{\K'}} R(M,\widetilde{F})<R(M^*,F^*)$, for all $\K'$ such that  $K\cup K'\in \K'$.  This will imply that the revenue guarantee will be strictly lower for any selling mechanism that bundles the groups $K$ and $K'$.  

 We construct $\widetilde{F}$ as follows. Without loss, assume $\ell_{K}>\ell_{K'}$ and let $\ell\in(\ell_{K'},\ell_{K})$. Let $H_{\varepsilon}$ be given by:
	\begin{align*}
		H_{\varepsilon}(x):=
		\begin{dcases}
			H(x)&x\le \ell-\varepsilon\\
			H(\ell-\varepsilon)&x\in(\ell-\varepsilon,\ell),\\
			H(x-\varepsilon)&x\ge \ell.
		\end{dcases}
	\end{align*}
 First, we define two parameters $\alpha^{\varepsilon}$ and $\ell^{\varepsilon}$ based on $\ell$ and $\varepsilon$:
 
\begin{align}
	&\int_{1}^{ \ell^{\varepsilon}}(\alpha^{\varepsilon}x)H_{\varepsilon}(\d x)+(\alpha^{\varepsilon} \ell^{\varepsilon})(1-H_{\varepsilon}( \ell^{\varepsilon}))=m_{K};\label{eqn:prop:2:1}\\
	&\int_{1}^{ \ell^{\varepsilon}}\phi_{K}(\alpha^{\varepsilon}x)H_{\varepsilon}(\d x)+\phi_K(\alpha^{\varepsilon} \ell^{\varepsilon})(1-H_{\varepsilon}( \ell^{\varepsilon}))=s_{K}.\label{eqn:prop:2:2}
\end{align}
Denote the LHS of \Cref{eqn:prop:2:1,eqn:prop:2:2} by $f_K^1(\alpha^{\varepsilon},\ell^{\varepsilon},\varepsilon)$ and $f_K^2(\alpha^{\varepsilon},\ell^{\varepsilon},\varepsilon)$, respectively. When $\varepsilon$ is sufficiently small, $\ell^{\varepsilon}$ is close to $ \ell_K$ and is thus  strictly larger than $\ell$.  

Therefore, we compute the Jacobian matrix of the functions   $\bm{f}_K:=(f_K^1, f_K^2)$ with respect to $(\alpha^{\varepsilon}, \ell^{\varepsilon})$:
\begin{align*}
	&\bm{J}_{\alpha^{\varepsilon}, \ell^{\varepsilon}}\bm{f}_K(\alpha^{\varepsilon}, \ell^{\varepsilon},\varepsilon)\\
	=&
	\left[\begin{matrix}
		\int_{1}^{\ell^{\varepsilon}}xH_{\varepsilon}(\d x)+\ell^{\varepsilon}(1-H_{\varepsilon}(\ell^{\varepsilon}))&\alpha^{\varepsilon}(1-H_{\varepsilon}(\ell^{\varepsilon}))\\
		\int_{1}^{\ell^{\varepsilon}}\phi_K'(\alpha^{\varepsilon}x)xH_{\varepsilon}(\d x)+\phi_K'(\alpha^{\varepsilon}\ell^{\varepsilon})\ell^{\varepsilon}(1-H_{\varepsilon}(\ell^{\varepsilon}))&\phi_K'(\alpha^{\varepsilon}\ell^{\varepsilon})\alpha^{\varepsilon}(1-H_{\varepsilon}(\ell^{\varepsilon}))
\end{matrix}\right] \\
&\implies\bm{J}_{\alpha^{\varepsilon}, \ell^{\varepsilon}}\bm{f}_K(\alpha^{\varepsilon}, \ell^{\varepsilon},\varepsilon)\big|_{\varepsilon=0}\\
	=&
	\begin{bmatrix}
		\int_{1}^{\ell_K}xH(\d x)+\ell_K(1-H(\ell_K))&\alpha_K(1-H(\ell_K))\\
		\int_{1}^{\ell_K}\phi_K'(\alpha_Kx)xH(\d x)+\phi_K'(\alpha_K\ell_K)\ell_K(1-H(\ell_K))&\phi_K'(\alpha_K\ell_K)\alpha_K(1-H(\ell_K))
	\end{bmatrix};\\
	\end{align*}
	Meanwhile, the partial derivative of $\bm{f}_K$ with respect to $\varepsilon$ is:
	\begin{align*}
	&\bm{J}_{\varepsilon}\bm{f}_K(\alpha^{\varepsilon},\ell^{\varepsilon},\varepsilon)\\
	=&\left[\begin{matrix} \alpha^{\varepsilon}\varepsilon h(\ell-\varepsilon)+\alpha^{\varepsilon}\int_{1}^{\ell^{\varepsilon}}h(x+\varepsilon)\d x\\
	(\phi_K(\alpha^{\varepsilon}\ell)-\phi_K(\alpha^{\varepsilon}(\ell-\varepsilon)))h(\ell-\varepsilon)+\alpha^{\varepsilon}\int_{1}^{\ell^{\varepsilon}}\phi_K'(\alpha^{\varepsilon}x)h(x+\varepsilon)\d x\end{matrix}\right]\\
	\implies
	&\bm{J}_{\varepsilon}\bm{f}_K(\alpha^{\varepsilon},\ell^{\varepsilon},\varepsilon)\big|_{\varepsilon=0}\\
	=&
	\begin{bmatrix}
		\alpha_K(H(\ell_K)-H(1))\\
		\alpha_K\int_{1}^{\ell_K} \phi'_K(\alpha_K x)H(\d x)
	\end{bmatrix}.
\end{align*}

By the inverse function theorem,
\begin{align*}
	\frac{\d \alpha^{\varepsilon}}{\d \varepsilon}\Big|_{\varepsilon=0}=-\bm{J}_{\alpha^{\varepsilon},\ell^{\varepsilon}}\bm{f}_K^{-1}\cdot \bm{J}_{\varepsilon}\bm{f}_K\Big|_{\varepsilon=0}=-\frac{\int_{1}^{\ell_K}(\phi'(\alpha_Kx)-\phi'(\alpha_K\ell_K))H(\d x)}{\int_{1}^{\ell_K}(\phi'_K(\alpha_Kx)-\phi'_K(\alpha_K\ell_K))xH(\d x)}<0,
\end{align*}

Therefore, for $\varepsilon$ sufficiently close to $0$, $\alpha^{\varepsilon}<\alpha_K$. Let $X$ be the random variable distributed according to CDF $H$. Define $\widetilde{\bm{V}}:=(\widetilde V_1, ..., \widetilde V_n)$, where
\begin{align*}
\widetilde	V_i=
	\begin{dcases}
	\frac{m_i}{m_J}\min\left\{ \alpha_{J} X ,\alpha_J\ell_J \right\} &\text{if }i\in J\neq K\\[10pt]
	\frac{m_i}{m_K}\min\left\{ \alpha^{\varepsilon} (X+\varepsilon \mathbf{1}_{\{X>\ell-\varepsilon\}}) ,\alpha^{\varepsilon},\ell^{\varepsilon} \right\} &\text{if }i\in J= K.
\end{dcases}
\end{align*}
Note that by definition, $X+\varepsilon \mathbf{1}_{\{X>\ell-\varepsilon\}}$ is distributed according to $H_{\varepsilon}$. Let $\widetilde{F}$ be the distribution of $\widetilde{\bm{V}}$. 

	Now, consider any mechanism $M$ that bundles $K\cup K'$. $M$ can be written as $(\psi(x),\tau(x),\widetilde{\psi}(x),\widetilde{\tau}(x))$,  where $(\psi(x),\tau(x))$ is the allocation and the payment for items $i\not\in K\cup K'$ and $(\widetilde{\psi}(x),\widetilde{\tau}(x))$  is the allocation and the payment for items $i\in K\cup K'$, all as functions of $x$, the report of $X$. Note this formalism does not imply that the sales of items $i\not\in K\cup K'$ is separated from those of  $i\not\in K\cup K'$.   The envelope theorem implies
	\begin{align*}
		R(M,\widetilde{F})\le& \sup_{\psi(\cdot),\widetilde{\psi}(\cdot)}\Bigg( \sum_{i\not\in K\cup K'}\int \psi_i(x)(\widetilde{V}_i(x)-\widetilde{V}'_i(x)\frac{1-H(x)}{h(x)})H(\d x) \cr
												 &+\int \widetilde{\psi}(x)\sum_{i\in K\cup K'}(\widetilde{V}_i(x)-\widetilde{V}'_i(x)\frac{1-H(x)}{h(x)}) H(\d x)\Bigg)\cr
			\le & \sum_{J\in \K, J\neq K,K'} \alpha_J+ \sup_x\left( \sum_{i\in K\cup K'} \widetilde{V}_i(x)(1-H(x)) \right)\cr
			=& \sum_{J\in \K,J\neq K,K'} \alpha_J +\sup_x\underbrace{\left(\alpha_{K'}\min\left\{ x,\ell_{K'} \right\}(1-H(x))+ \alpha^{\varepsilon}\min\left\{ x,\ell^{\varepsilon}  \right\}(1-H_{\varepsilon}(x)) \right)}_{A(x)}.
	\end{align*}
	The second inequality is from the definition of $\alpha_J$'s and the fact that $\widetilde{\psi}$ equivalently characterizes a mechanism that bundles $K\cup K'$. For $x\le \ell-\varepsilon$, $H_{\varepsilon}(x)=H(x)$, but $\alpha^{\varepsilon}<\alpha_K,$ so $$A(x)= \alpha_{K'}+\alpha^{\varepsilon}<\alpha_{K'}+\alpha_{K}.$$  For $x\in(\ell-\varepsilon,\ell^{\varepsilon}]$, $x>\ell_{K'}$ when $\varepsilon$ is chosen sufficiently small. Therefore,
	\begin{align*}
		A(x)=&\alpha_{K'}\frac{\ell_{K'}}{x}+\alpha^{\varepsilon}\frac{x}{x-\varepsilon}\\
		<&\alpha_{K'}\frac{\ell_K}{\ell-\varepsilon}+\alpha^{\varepsilon}\frac{\ell-\varepsilon}{\ell-2\varepsilon}.
	\end{align*}
	As  $\varepsilon\to 0$, the latter expression tends to $ \alpha_K+\alpha_{K'}\frac{\ell_K}{\ell}<\alpha_K+\alpha_{K'}$. Combining both case proves that when $\varepsilon$ is sufficiently small, $R(M,\widetilde{F})<\sum \alpha_J=R(M^*,F^*)$. Therefore,
\begin{align*}
	\sup_{M\in\M_{\K}}R(M,\tilde{F})<R(M^*,F^*),
\end{align*} 	
as was to be shown.

\subsection{Proof of \Cref{thm:info}} \label{sec:proof-dr}

\begin{proof}  Fix any prior $G\in\Delta(\mathbb{R}^n_+)$ that is stochastically monotonic with sharing parameters $\bm{\xi}=(\xi_1, ..., \xi_n)\in \mathbb{R}^n_{+}$ with $\sum_i \xi_1=1$. We prove the theorem by construction and verification.  We would like to invoke \Cref{thm:moment} for the case of $\K=\{N\}$. We begin by identifying a relaxed ambiguity set. For all $\bm{v}\in \R^n_+$, let $\bm{v}=\left( v_1, ..., v_n \right)$.  For  any $z>0$, consider an ambiguity set
	\begin{align*}
		\F_{z}=\left\{ F\in \Delta (\R^n_+)|\E_F[\bm{v}]=\E_{G}[\bm{v}]\text{ and }\E_F[\phi_{z}(\bm{v})]\le\E_{G}[\phi_{z}(\bm{v})] \right\},
	\end{align*}
	indexed by  ${z}$, 
	where  $\phi_{z}( \bm{v}):=\max\left\{ z-\sum_{i}v_i,0 \right\}$.   More importantly, the associated condition captures the second-order stochastic dominance order with respect to the total sum of item values:  namely, 
the random variable $\sum_{i} v_i$ distributed according to $F$ second-order stochastically dominates  (SOSD) the corresponding random variable  distributed according to $G$ if $\E_F[\phi_{z}(\bm{v})]\le\E_{G}[\phi_{z}(\bm{v})]$ for all $z> 0$.\footnote{This can be seen upon integration by parts:
\begin{align*}
	\int_{\sum_{i} v_i \le z} F(\bm{v})\d \bm{v}=\int_{\sum_{i} v_i\le z} (z-\sum v_i)F(\d \bm {v})=\E_F[\phi_{z}(\bm{v})] \le\E_{G}[\phi_{z}(\bm{v})]= \int_{\sum_{i} v_i\le z} G(\bm{v})\d \bm{v},
\end{align*}
where the  inequality is from $F\in \F_{z}$.}  

Observe that $\phi_{z}$ is not twice-differentiable, as is required by \Cref{thm:moment}. Nevertheless, it is piecewise linear with only one kink; hence, the saddle point can be explicitly constructed following the same procedure as in \Cref{thm:moment}, denoted by $(F_{z},M_{z})$, where $M_{z}$ is a pure bundling mechanism. Let $\alpha_z$ be the parameter defining $F_z$, indexed by $z$. Then, let $z^*$ minimizes $\alpha_z$ while maintaining
\begin{align*}
    \E_{F_{\alpha_z}}[\phi_{z'}(\bm{v})]\le \E_G[\phi_{z'}(\bm{v})]
\end{align*}
holds for all $z'\in\R_+$. Define $F^*:=F_{z^*},\ M^*:=M_{z^*}$.\footnote{The formal construction of the saddle point and the identification of $z^*$ is relegated to \Cref{ssec:saddle}.} 

We are now ready to prove our statement: $(M^*,F^*)$ is a saddle point under ambiguity set $\F_G$, defined in \cref{eq:F-info}.  To this end, note first that $\F_G\subset\bigcap_{ {z}>0}\F_{ {z}}\subset \F_{z^*}$ (since $(\phi_z)_{z>0}$ comprises only a subset of all possible convex functions). In light of  \Cref{cor:beyond}, it suffices to show $F^*\in \F_G$, since that will imply that $M^*$ is maxmin optimal given ambiguity set $\F_G$.

To show $F^*\in \F_G$, fix any arbitrary convex function $\phi: \mathbb{R}_+^n\to \mathbb{R}$.  Let $\bm{m}=(m_1,...,m_n)$, where $m_i:=\E_G[v_i]=\E_G[\xi_i(\tsum_j v_j)]$.
\begin{align*}
	\E_{G}[\phi(\bm{v})]=& \, \E_{G}\Big[ \E [ \phi\left(\bm{v} \right) | \tsum_i v_i]   \Big]\\
	\ge& \, \E_{G}\left[  \phi\big( \E[ {\bm{v}} |\tsum_i v_i] \big)   \right]\\
	=& \,\E_{G}\big[ \phi\left( \bm{\xi}\cdot \tsum v_i   \right) \big]\\
 	\ge& \, \E_{F^*}\big[ \phi\left( \bm{\xi}\cdot  \tsum v_i   \right) \big]\\
	=& \,\E_{x \sim H}\left[ \phi\left( \min\left\{ \alpha_N x,\alpha_K e^{\frac{\sum m_i-\alpha_N}{\alpha_N}}\right\}\cdot \frac{\bm{m}}{\sum m_i}  \right) \right]\\
	=& \, \E_{F^*}\left[ \phi(\bm{v}) \right].
\end{align*}
The first inequality follows from the convexity of $\phi$.  The second equality follows from stochastic comonotonicity of $G$; i.e., $G$ satisfies \Cref{as:comon} for some $\bm{\xi}=(\xi_1, ..., \xi_n)\in \mathbb{R}^n_{+}$ with 
$\sum_i \xi_1=1$.  To understand the second inequality, observe first that 
the composition function $\phi\circ \bm{\xi}$ is convex in $\sum_i v_i$, the total sum  of all item  values. Next, recall that  the total sum of all item  values distributed according to  $F^*$  SOSD the total sum of item values distributed according to $G$---a fact implied by $F^*\in\bigcap_{{z}}\F_{{z}}$.  Hence, the second inequality follows.  The last two equalities follow from the definition of $F^*$ given in \Cref{sec:moment}. Since $\phi$ is an arbitrary convex function, we have shown that $F^*\in \F_G$, and the proof of \Cref{thm:info} is complete. \end{proof}

\subsection{Proof of \Cref{thm:partial:bundling}}\label{sec:proof-saddle}

\begin{proof}  Observe 
	\begin{align*}
	\sup_{M\in\M_{\K}}\inf_{F\in\F\bigcap \widehat{\F}}R(M,F)=&\inf_{F\in\F\bigcap \widehat{\F}}\sup_{M\in\M_{\K}} R(M,F)\\
	=&\inf_{\left( F_K \right)\in \Upsilon_{\K}(\F)}\inf_{F\in\Upsilon^{-1}_{\K}((F_K))\bigcap\widehat{\F}}\sup_{M\in\M_{\K}}R(M,F)\\
	\ge& \inf_{\left( F_K \right)\in \Upsilon_{\K}(\F)}\sup_{M\in\M_{\K}}\inf_{F\in\Upsilon^{-1}_{\K}((F_K))\bigcap\widehat{\F}}R(M,F)\\
	=& \inf_{\left( F_K \right)\in \Upsilon_{\K}(\F)}\sup_{M\in\M}\inf_{F\in\Upsilon^{-1}_{\K}((F_K))\bigcap\widehat{\F}}R(M,F)\\
	\ge&\sup_{M\in\M}\inf_{\left( F_K \right)\in \Upsilon_{\K}(\F)}\inf_{F\in\Upsilon^{-1}_{\K}((F_K))\bigcap\widehat{\F}}R(M,F)\\
	=&\sup_{M\in\M}\inf_{F\in\F\bigcap \widehat{\F}}R(M,F).
	\end{align*}
	The first equality follows from \Cref{lem:minimax} proven in \Cref{sec:technical}, where $\G=\F\bigcap\widehat{\F}$. The two inequalities are min-max inequalities. The third equality follows from \Cref{prop:saddle}. Since $\M_{\K}\subset \M$, the above   inequalities yields the desired statement. 
\end{proof}

\begin{lemma}	
	\label{prop:saddle} Fix any $\K$-marginals 
	$( F_K)_{K\in\K}$ and any $\widehat{\F}$, and let $\F:=\Upsilon_{\K}^{-1}(( F_K)_{K\in\K} )\bigcap \widehat{\F}\neq \emptyset$.  Then,  $\K$-bundled sales is robustly optimal in the sense that 
	\begin{align*}
	\sup_{M\in\M_{\K}}\inf_{F\in\F}R(M,F)
	=\sup_{M\in\M}\inf_{F\in\F}R(M,F).	\end{align*}
\end{lemma}

\begin{proof}  We first construct the worst case $F^*\in \Delta(\R^n_+)$.  To this end, we imagine a hypothetical problem in which the seller sells $k=|\K|$ goods and faces full ambiguity given the knowledge of the marginal distributions $ F_K$ of the values of each item $K\in \K$. (That is, we interpret bundle $K$ as a single item in this hypothetical problem.) This is precisely what \cite{carroll2017robustness} analyzed. To recast his result in the current setup for this hypothetical problem, let $\M^{k}$ be the set of feasible mechanisms in this hypothetical problem with $k$ items, and $R^{k}(M,G)$ denote the revenue the seller collects from a mechanism $M\in \M^{k}$ facing distribution $G\in \Delta(\R^{k}_+)$. (The corresponding notations for our original problem would then have superscript $n$, which we suppress for convenience.) Consider $\M^{k}_{\mathcal H}$, where $\mathcal H$ is the finest partition of $\K$. Then $\M^k_{\mathcal{H}}$ is the set of all ``separate sales'' mechanisms in this hypothetical problem.

	Theorem 1 of \cite{carroll2017robustness} then proves that there exists $G^*\in \Upsilon^{-1}_{\mathcal H}((F_K)_{K\in\K})\subset \Delta(\R^{k}_+)$ and 
	\begin{align} \label{eq:carroll}
	\sup_{M\in \M^k_{\mathcal{H}}}R^{k}(M,G^*)=\sup_{M\in\M^{k}}R^{k}(M,G^*).
	\end{align}
	
	Now we construct $F^*\in\Delta(\R^{n}_+)$ using $G^*\in \Delta(\R^{k}_+)$. Choose any  $F'\in \F=\Upsilon^{-1}_K(( F_K))\bigcap \widehat{\F}$ (which we assumed to be nonempty). For each $i$, let $$\alpha_i:=\frac{\E_{F'}[v_i]}{\sum_{j\in K(i)} \E_{F'}[v_j]}.$$
	 Let $\bm{X}$ be a $k$-dimensional random vector distributed according to $G^*$, and let $\bm{V}=(V_i)$ be a $n$-dimensional random vector defined by
	 $$V_i(\bm{X}):=\alpha_iX_{K(i)},$$ for each $i$.  Let $F^*$ be the distribution of $\bm{V}$.  
	
	We now prove $F^*\in \Upsilon^{-1}_K(( F_K))\bigcap \widehat{\F}$. Since $G^*\in\Upsilon^{-1}_{\mathcal H}((F_K))$,  by construction,  $F^*\in \Upsilon^{-1}_K(( F_K))$.
	Since $F'\in \Upsilon^{-1}_K(( F_K))$,
	$$\P_{F'}\left\{\tsum_{j\in K}v_j\le y\right\}= \P_{G^*}\left\{X_K\le y\right\}= \P_{F^*}\left\{\tsum_{j\in K}v_j\le y\right\},$$ for each $K\in \K$ and $y\in \R_+$.  Hence, $\E_{F'}[\sum_{j\in K}v_j]=\E_{G^*}[X_K]=\E_{F^*}[\sum_{j\in K}v_j]$.  It further follows that $\E_{F'}[v_j]=\E_{F^*}[v_j]$.  Hence, $F^*\in   \widehat{\F}$.  We thus conclude that $F^*\in \F$.

	\par
	We next prove that 	
	\begin{align} \label{eq:MGF}
	\sup_{M\in \M^{k}}R^{k}(M,G^*)\ge \sup_{M\in\M}R(M,F^*).
	\end{align}
	To see this, fix any mechanism $M=(q,t)\in\M$ in our original problem.  We now construct another mechanism $\tilde M=(\tilde{q},\tilde{t})\in \M^k$ for the hypothetical $k$-item problem as follows:     
	\begin{align*}
	\begin{dcases}
	\tilde{q}_K(\bm{x})=\sum_{j\in K}\alpha_jq_j(\bm{V}(\bm{x})),\\
	\tilde{t}(\bm{x})=t(\bm{V}(\bm{x})).
	\end{dcases}
	\end{align*}
	Observe that  $\forall \bm{x},\bm{x}'$,
	\begin{align*}
	\bm{x}\cdot \tilde{q}(\bm{x}')-\tilde{t}(\bm{x}')
	=&\sum_{K\in \K} x_K \sum_{j\in K}\alpha_j q_j(\bm{V}(\bm{x}'))-t(\bm{V}(\bm{x}'))\\
	=&\sum_{i} \frac{V_i(\bm{x})}{\alpha_i} \alpha_i q_i(\bm{V}(\bm{x}'))-t(\bm{V}(\bm{x}'))=\bm{V}(\bm{x})\cdot q(\bm{V}(\bm{x}'))-t(\bm{V}(\bm{x}')).
	\end{align*}
	Therefore, $(IC)$ and $(IR)$ of $M$ on $\mbox{supp}(F^*)$ imply $(IC)$ and $(IR)$ of $\tilde{M}$ on $\mbox{supp}(G^*)$.  Hence, $\tilde M \in \M^{k}$. 	
	Moreover, given $G^*$, $\tilde{M}$ yields  the same expected revenue as $M$ given $F^*$.   Since one can find such $\tilde M$ for each $M\in \M$,  \cref{eq:MGF} follows.

	We next prove 
	\begin{align} \label{eq:SPB} 
	\sup_{M\in\M_{\K}}R(M,F^*)\ge \sup_{M\in\M^k_{\mathcal{H}}}R^{k}(M,G^*).
	\end{align}
	Indeed, for any  $\tilde{M}\in\M^k_{\mathcal{H}}$, we can construct a $\K$-bundled sales mechanism $\widehat M=(\widehat q, \widehat t)\in \M_{\K}$, where 
	\begin{align*}
		\widehat{q}_i(\bm{v}):=\tilde{q}_{K(i)}(\tsum_{j\in K(i)}v_j) \mbox{ and }
		\widehat{t}(\bm{v}):=\sum_{K\in \K}\tilde{t}_{K}(\tsum_{j\in K}v_j).
	\end{align*}
	Whenever $\bm{v}=\bm{V}(\bm{x})$, $\sum_{j\in K}v_j=\sum_{j\in K} \alpha_j x_K=x_K$. So, $\widehat{M}$ given $F^*$ is payoff equivalent to $\tilde{M}$ given $G^*$.  Since $\tilde M\in \M_{\mathcal H}^k$, satisfying $(IR)$ and $(IC)$  on $\mbox{supp}(G^*)$, $\widehat M$ satisfies $(IC)$ and $(IR)$  on $\mbox{supp}(F^*)$.  
	
	Combining \cref{eq:SPB}, \cref{eq:carroll}, and \cref{eq:MGF},  we obtain 
	\begin{align} \label{eq:MMK}
\sup_{M\in\M_{\K}}R(M,F^*)\ge \sup_{M\in\M^k_{\mathcal{H}}}R^{k}(M,G^*)=\sup_{M\in\M^{k}}R^{k}(M,G^*)\ge  \sup_{M\in\M}R(M,F^*).
	\end{align}

Finally, fix any mechanism $M\in \M_{\K}$. For any  $F\in\F=\Upsilon^{-1}_K((F_K))\bigcap \widehat{\F}$,
\begin{align*}
R(M,F)=&\int t(\bm{v})F(\d \bm{v})
=\sum_{K\in \K} \int t(\tsum_{j\in K} v_j)F(\d \bm{v})
=\sum_{K\in \K}\int t(x)F_K(\d x),
\end{align*}
where the first equality follows from the fact that $M\in \M_{\K}$ and the second follows from $F\in\Upsilon^{-1}_K((F_K))$.  In other words, $R(M, F)=R(M, F')$ for any $F,F'\in \Upsilon^{-1}_K((F_K))\bigcap \widehat{\F}=\F$, as long as $M\in M_{\K}$.  Hence, it follows that
\begin{align}
\sup_{M\in\M_{\K}}\inf_{F\in\F} R(M,F)=	\sup_{M\in\M_{\K}}R(M,F^*).  \label{eq:invF}
\end{align}
Combining  \cref{eq:invF} with \cref{eq:MMK}, we get
$$\sup_{M\in\M_{\K}}\inf_{F\in\F} R(M,F)= \sup_{M\in\M_{\K}}R(M,F^*)\ge \sup_{M\in\M}R(M,F^*)\ge \sup_{M\in\M}\inf_{F\in\F} R(M,F).$$
Since $\M_{\K}\subset\M$, the reverse inequality also holds, so we have the desired conclusion. \end{proof}

\newpage
\section{Online Appendix: Not for Publication}\label{sec:supp}

	\subsection{Supplemental Results for \Cref{thm:moment}}  \label{sec:ift}
Here, we verify \cref{eqn:computation}. Let
\begin{align}
	&\begin{dcases}
		m_K'=\sum_{j\in K}\E_F[v_j]\\
		s_K'=\int \phi_K(\textstyle{\sum_{j\in K}} (v_j)  F(\d \bm{v})
	\end{dcases}\cr
	\implies&
	\begin{dcases}
		\frac{\d m_K'}{\d \delta}\Big|_{\delta=0}=\sum_{j\in K}(\E_{F}[v_j]-\E_{F^*}[v_j])\\
		\frac{\d s_K'}{\d \delta}\Big|_{\delta=0}= \int\phi_K(\mbox{$\sum_{j\in K}$}v_j)F(\d \bm{v})-s_K\label{eqn:FOC:delta}
	\end{dcases}
\end{align}
Straightforward algebra yields:
{\medmuskip=0mu
\thinmuskip=0mu
\thickmuskip=0mu
\small
		\begin{align}
			\frac{\d B}{\d \delta}\Big|_{\delta=0}=&\sum_{K\in\K}\lambda_K\Big( -\phi'_K(\beta_K)\Big(\sum_{j\in K}(\E_{F}[v_j]-\E_{F^*}[v_j]) \Big)+\int \phi_K(\mbox{$\sum_{j\in K}$}v_j)F(\d \bm{v})-\int \phi_K(\mbox{$\sum_{j\in K}$}v_j)F^*(\d \bm{v}) \Big)\cr
			=&-\sum_{K\in\K} \lambda_K \phi_K'(\beta_K)\cdot \frac{\d m_K'}{\d \delta}\Big|_{\delta=0}+\sum_{K\in \K} \lambda_K \cdot \frac{\d s_K'}{\d \delta}\Big|_{\delta=0}.  \label{eqn:FOC:R}
	\end{align}}
	We show that \cref{eqn:FOC:R} coincides with $-\frac{\d \sum_{K\in\K}\alpha_K}{\d \delta}\Big|_{\delta=0}$.  To this end, we first apply the implicit function theorem to \cref{eq:mi':1,eq:ki:1}. To begin, define function $\psi$:
	\begin{align*}
	\psi(\alpha,\beta,m,s)=\left\{ 
	\begin{array}{c}
	\int_{\alpha}^{\beta}\dfrac{\alpha}{x}\d x+\alpha-m\\
	\int_{\alpha}^{\beta}\dfrac{\alpha\phi(x)}{x^2}\d x+\frac{\phi(\beta)\alpha}{\beta}-s
	\end{array}
	\right\}.
	\end{align*}
	Then by \Cref{lem:alphabeta}, $(\alpha_K, \beta_K)$ is the unique solution to $\psi(\alpha_K,\beta_K,m_K,s_K)=0$. 
	\begin{align*}
	&\frac{\d \psi}{\d (\alpha,\beta)}= \left[ \begin{array}{cc}
	\log(\frac{\beta}{\alpha})&\frac{\alpha}{\beta}\\
	\int_{\alpha}^{\beta}\frac{\phi'(x)}{x} \d x & \frac{\alpha \phi'(\beta)}{\beta}
	\end{array} \right],
		\end{align*}
	and
	$$ \mathrm{det}\left( \frac{\d \psi}{\d (\alpha,\beta)} \right)=\frac{\alpha}{\beta}\int \frac{\phi'(\beta)-\phi'(x)}{x}\d x>0,$$
	so  $\frac{\d \psi}{\d (\alpha,\beta)}$ is invertible at $(\alpha_K,\beta_K)$. The implicit function theorem then implies that $(\alpha_K,\beta_K)$ is locally differentiable in $m$.  With the subscript $K$ suppressed for notational convenience, we have: 
	\begin{align*}
		\frac{\d (\alpha,\beta)^T}{\d (m,s)}=&-\left[\frac{\d \psi}{\d (\alpha, \beta)}\right]^{-1}\cdot \frac{\d \psi}{\d (m,s)}\\
	=&-\left[ \begin{array}{cc}
	\log(\frac{\beta}{\alpha})&\frac{\alpha}{\beta}\\
	\int_{\alpha}^{\beta}\frac{\phi'(x)}{x}\d x & \frac{\alpha \phi'(\beta)}{\beta}
	\end{array} \right]^{-1}\cdot\left[ \begin{array}{cc}
	-1&0\\
	0&-1
	\end{array}\right].
	\end{align*}
	It follows that: 
	\begin{align}
		\begin{dcases}
			\frac{\d \alpha_K}{\d m_K'}\Big|_{m_K'=m_K}=\lambda_K\phi'_K(\beta_K);\\
			\frac{\d \alpha_K}{\d s_K'}\Big|_{s_K'=s_K}=-\lambda_K.
\end{dcases}\label{eqn:FOC:m}
	\end{align}
	Combining  \cref{eqn:FOC:m,eqn:FOC:R},  we conclude that
	\begin{align*}
		\frac{\d \sum_{K\in\K}\alpha_K}{\d \delta}\Big|_{\delta=0}=&\sum_{K\in\K} \left(\frac{\d \alpha_K}{\d m_K'}\cdot \frac{\d m_K'}{\d \delta}\Big|_{\delta=0}+\frac{\d \alpha_K}{\d s_K'}\cdot \frac{\d s_K'}{\d \delta}\Big|_{\delta=0}\right)=-\frac{\d B}{\d \delta}\Big|_{\delta=0},
	\end{align*}
	which yields \cref{eqn:computation}.

\subsection{Technical Lemma used in the proof of \Cref{thm:partial:bundling}}\label{sec:technical}

\begin{lemma}		\label{lem:minimax}
	If $\G$ is regular, then
	\begin{align}
		\inf_{F\in\G}\sup_{M\in\M_{\K}}R(M,F)=\sup_{M\in\M_{\K}}\inf_{F\in\G}R(M,F).\label{eq:infsup0}
	\end{align}
\end{lemma}

\begin{proof}
	In this proof, we still adopt the notation in  the proof of \Cref{prop:saddle} whereby $\M^1$ denotes the set of feasible mechanisms for selling a single item. By the definition of $\M_{\K}$, there exists $(q_K,t_K)\in \M^1$ such that $t(\bm{v})\equiv \sum_{K\in\K}t_K(\sum_{j\in K}v_j)$. Let $\mathcal{T}=\big\{ t\in\R_+^{\R_+}:\exists M=(q,t)\in \M^1 \big\}$ denote the projection of $\M^1$ onto the payment dimension. Then, $M\in \M_{\K}$ if and only if $(t_K)\in \mathcal{T}^{|\K|}$ such that  $\forall F\in \Delta(\R^n_+)$:
\begin{align*}
	R(M,F)= \int  \sum_{K\in\K}t_K(\tsum_{j\in K}v_j)  F(\d \bm{v}).
\end{align*}
 Therefore, \cref{eq:infsup0} is equivalent to
\begin{align}
	\inf_{F\in\G}\sup_{(t_K)\in \mathcal{T}^{|\K|}} \int  \sum_{K\in\K}t_K(\tsum_{j\in K}v_j)  F(\d \bm{v}) =\sup_{(t_K)\in \mathcal{T}^{|\K|}} \inf_{F\in\G} \int  \sum_{K\in\K}t_K(\tsum_{j\in K}v_j)  F(\d \bm{v}).\label{eq:infsup}
\end{align}

Let $\mathcal{T}_c:=\mathcal{T}\bigcap \mathcal{BC}(\R_+)$, where $ \mathcal{BC}(\R_+)$ is the set of all bounded and continuous functions defined on $\R_+$.  We first establish the minimax theorem within $\mathcal{T}_c^{|\K|}$:
	\begin{claim} \label{claim:minmax}
		$\min\limits_{F\in\G}\sup\limits_{(t_K)\in\mathcal{T}_c^{|\K|}}\int \sum_{K\in\K}t_K(\tsum_{j\in K} v_j)F(\d \bm{v})=\sup\limits_{(t_K)\in\mathcal{T}_c^{|\K|}}\min\limits_{F\in\G}\int \sum_{K\in\K}t_K(\tsum_{j\in K} v_j)F(\d \bm{v})$.
	\end{claim}
	\begin{proof}  We observe that
		\begin{enumerate}
			\item $\mathcal{T}_c^{|\K|}$ is a convex subset of linear topological space $\mathcal{C}(\R_+)^{|\K|}$ (equipped with sup norm);
			\item $\G$ is a compact and convex subset of a linear topological space $\Delta(\R_+^n)$ (equipped with L\'{e}vy-Prokhorov metric);
			\item $R(M,G)=\int \tsum_{K\in\K}t_K(\tsum_{j\in K}v_j)F(\d \bm{v})$ is linear and continuous in both $(t_K)$ and $F$.
		\end{enumerate}
		
		Since every mechanism in $\M^1$ is incentive compatible,   for any $t_K\in\mathcal{T}_c$, there exists an nondecreasing function $q_K:\R_+ \to [0,1]$, such that $t_K({v})=v q_K({v})-t_K({0})-\int_{{0}}^{{v}}q_K(x)\d x$ for each $v\in \R_+$.   The convexity of $\mathcal{T}_c$ then follows easily: since any convex combination of nondecreasing functions $q_K$ and $q_K'$ is still nondecreasing, a convex combination of  $t$ and  $t'$ in $\mathcal{T}_c$ is still an element of $\mathcal{T}_c$. Then $\mathcal{T}_c^{|\K|}$ is convex since it is the product space. Next, the compactness of $\mathcal{F}$ follows from Prokhorov theorem, since $\mathcal{F}$ is tight and  closed.  The linearity of $\int \sum_{K\in\K}t_K(\sum_{j\in K}v_j)F(\d \bm{v})$ is straightforward. The continuity follows since $t_K$ is bounded and continuous and $F$ is a probability measure (Portmanteau theorem). The three observations ensure that the Sion's minmax theorem (\cite{sion1958general}) holds, from which the claim follows.
	\end{proof}\par

	\begin{claim} \label{claim:approx} Fix any $F\in \G$ and  $(t_K)\in \mathcal{T}^{|\K|}$. For any $\varepsilon>0$, there exists $(\tilde t_K)\in \mathcal{T}^{|\K|}_c$ such that 
 \begin{align}
 \int  \sum_{K\in\K}\tilde t_K({\tsum_{j\in K}}v_j)  F(\d \bm{v})
 \ge \int \sum_{K\in\K}  t_K(\tsum_{j\in K}v_j)  F(\d \bm{v})-\epsilon.  \label{eqn:continuity}
 \end{align}
\end{claim}
 
 \begin{proof} For each $t\in \mathcal{T}$,  there exists a nondecreasing function $q$, such that $t({v})=v q({v}) -t({0})-\int_{{0}}^{{v}}q(s)\d s$ for  each $v\in \R^+$. Fix any such  $t\in\mathcal{T}$ and the associated $q$. For each $\varepsilon$, we construct   a  continuous function  $\widehat{q}: \R_+\to [0,1]$ such that  $\widehat{q}\ge q-\varepsilon$ and $\int_0^{v} \widehat{q}(x)\d x- \int_0^{v} q(x)\d x \le\varepsilon$ for each $v\in \R_+$.
	
	Since $q$ is nondecreasing, there exist countably many discrete points, $\left\{ v^m \right\}$, at which $q$ jumps up. Let $d^m>0$ be the size of the jump at $v^m$, respectively. Define $r^m(v)=\pmb{1}_{\{q(v)\ge q(v^{m+})\}}d^m$. $r^m$ is a step function with step size $d^m$ at $v^m$. Then, $w(v)=q(v)-\sum r^m(v)$ is a nondecreasing and  continuous function.  Recall $q(\cdot)\le 1$, so   $\sum_m ||r^m||=\sum_m d^m<\infty$, and for any  $\varepsilon>0$,  there exists $N$ such that $\sum_{m>N}||r^m(v)||\le \varepsilon$, where $ ||\cdot||$ is the supnorm. 
	
	Define:
	\begin{align*}
	\widehat{r}^m(v):=
	\begin{dcases}
	0&\text{when }v\le v^{m}-\frac{\varepsilon}{N},\\
	\frac{Nd^m}{\varepsilon}\left(v-v^{m}+\frac{\varepsilon}{N}\right)&\text{when }v\in\left( v^m-\frac{\varepsilon}{N},v^m \right),\\
	d^m&\text{when }v\ge v^m.
	\end{dcases}
	\end{align*}
	In words, $\widehat{r}^m(v)$ approximates the step function $r^m(v)$ using a continuous piecewise linear function. By definition, $\widehat{r}^m(v)$ is nondecreasing and continuous.  Moreover,  $\widehat{r}^m(v)\ge r^m(v)$ and $\int_0^v|r^m(s)-\widehat{r}^m(s)|\d s\le \frac{\varepsilon d^m}{N}$. Now define $\widehat{q}(v):=w(v)+\sum_{m=1}^{N}\widehat{r}^m(v)$. Then, by definition,
	\begin{enumerate}
		\item $\widehat{q}(v)$ is nondcreasing and continuous in $v$;
		\item $\widehat{q}(v)-q(v)=\sum_{m=1}^N(\widehat{r}^m(v)-r^m(v)) -\sum_{m=N+1}^{\infty}r^m(v)\ge-\varepsilon$; 
		\item $\int_0^v [\widehat{q}(x)-q(x)]\d x=\int_0^v [\sum_{m=1}^N(\widehat{r}^m(x)-r^m(x)) -\sum_{m=N+1}^{\infty}r^m(x)]\d x\le N\times \frac{\varepsilon}{N}=\varepsilon $.
	\end{enumerate}
	Next, since $\widehat{q}(v)$ is nondecreasing and bounded, there exists $v^*$ such that  $\forall v\ge v^*$, $\widehat{q}(v)\le \widehat{q}(v^*)+\varepsilon$. Truncate $\widehat{q}$ by defining $\tilde{q}(v)=\widehat{q}(\min\left\{ v,v^* \right\})$. We have 
	\begin{align*}
	\tilde{t}(v):=&v\tilde{q}(v)-t(0)-\int_0^v\tilde{q}(x)\d x\\
	\ge&v(\widehat{q}(v)-\varepsilon)-t(0)-\int_0^{v}\widehat{q}(x)\d x\\
	\ge&v(q(v)-2\varepsilon)-t(0)-\left( \int_0^vq(x)\d x+\varepsilon \right)\\
	\ge&t(v)-(1+2v)\varepsilon,
	\end{align*}
	where the first inequality follows from the definition of $\tilde q$ and the second follows from observations 2 and 3 above. 
		Since $\tilde{q}(v)\equiv \tilde{q}(v^*)$ for $v\ge v^*$, $\tilde{t}(v)-\tilde{t}(v^*)=(v-v^*)\tilde{q}(v^*)-\int_{v^*}^v\tilde{q}(v^*)\d s=0$. So $\tilde{t}$ is continuous and bounded and hence $\tilde{t}\in \mathcal{T}_c$. \par

	For all $(t_K)\in \mathcal{T}^{|\K|}$, let $(\tilde{t}_K)$ be constructed as above with parameter $\varepsilon$. Then, $\forall F\in\mathcal{F}$,
	\begin{align*}
		\int \sum_{K\in\K}\tilde{t}_K(\tsum_{j\in K}v_j)F(\d \bm{v})\ge&\int \sum_{K\in\K}t_K(\tsum_{j\in K}v_j)F(\d \bm{v})-\left( 1+2\E_{F}[\tsum v_i] \right)\varepsilon\\
	\end{align*}
	Recall that $F$ has finite means for item values---an implication of $\F$ being regular.  Hence, by setting $\varepsilon:=  \epsilon/\left( 1+2\E_{F}[\tsum v_i] \right)$, the claim is proven.
\end{proof}

It follows from 
\Cref{claim:approx} that  
\begin{align*}
\sup_{(t_K)\in \mathcal{T}^{|\K|}_c}\int  \sum_{K\in\K}t_K(\tsum_{j\in K}v_j) F(\d \bm{v})\ge \sup_{(t_K)\in \mathcal{T}^{|\K|}} \int  \sum_{K\in\K}t_K(\tsum_{j\in K}v_j)  F(\d \bm{v}).  
\end{align*}
Since $\mathcal{T}^{|\K|}_c\subset \mathcal{T}^{|\K|}$, this in turn implies that 
 \begin{align}
\sup_{(t_K)\in \mathcal{T}^{|\K|}_c}\int  \sum_{K\in\K}t_K(\tsum_{j\in K}v_j) F(\d \bm{v})= \sup_{(t_K)\in \mathcal{T}^{|\K|}} \int  \sum_{K\in\K}t_K(\tsum_{j\in K}v_j)  F(\d \bm{v}).   \label{eqn:continuity}
 \end{align}

Applying \Cref{eqn:continuity} to \Cref{claim:minmax} establishes 
\Cref{eq:infsup}, which is in turn equivalent to \Cref{eq:infsup0}.    
\end{proof}

 \subsection{\Cref{prop:mechanism:extension} referred to in \Cref{fn:mechanism}}

\begin{proposition}
	Let $V\subset \R^n$ be a closed set, and $(q,t):V\to [0,1]^n\times \R_+$ be Borel measurable function that satisfy $(IC)$ and $(IR)$ on $V$. Then there exists Borel measurable function $(q^*,t^*)$ s.t. (i) $(q^*,t^*)\equiv (q,t)$ on $V$ and  (ii) $(q^*,t^*)$ satisfy $(IC)$ and $(IR)$ on $\R^n_+$.
	\label{prop:mechanism:extension}
\end{proposition}

\begin{proof}
	Let $U:=\R^n_+\setminus V$. Let $W:=\mbox{cl}(\{ (q(v ),t(v )): v\in V\})$, so $W$ is   closed.  For each $v\in U$, define:
	\begin{align*}
	(Q(v),T(v))=\left\{ (x,y)\in W: v\cdot x-y=\sup_{v'\in V} v\cdot q(v')-t(v')\right\}.
	\end{align*}
	In words, $(Q,T)$ is a correspondence that ``maximally'' extends $(q,t)$ to all types in  $\R^n_+$. Clearly, a mechanism using any selection from  $(Q(v),T(v))$ will satisfy $(IC)$ and $(IR)$ for all $v\in U$.   Now we verify that the correspondence $(Q(v),T(v))$ is Borel measurable on $U$.  $\forall$ closed set $D\in [0,1]^n\times \R_+$:
	\begin{align*}
	&\left\{ v\in U:(Q(v),T(v))\bigcap D\neq \emptyset  \right\}\\
	=&\left\{ v\in U: \max_{x,y\in D\bigcap W} v\cdot x-y=\sup_{v'\in V} v\cdot q(v')-t(v')  \right\}\\
	=&\left\{ v\in U: \sup_{x,y\in D\bigcap W} v\cdot x-y=\sup_{v'\in V} v\cdot q(v')-t(v')  \right\}\\
	=&U\setminus \left\{ v\in U: \sup_{x,y\in D\bigcap W} v\cdot x-y<\sup_{v'\in V} v\cdot q(v')-t(v')  \right\}
	\end{align*}
	The first and third equality are identities. The second equality is by $D\bigcap W$ being closed and $\forall v$, the set of $(x,y)$ s.t. $v\cdot x -y\ge0 $ is bounded. For any $  v\in U$ such that $\sup_{x,y\in\ D\bigcap W}v\cdot x-y<\sup_{v'\in V}v\cdot q(v')-t(v')$, let $\varepsilon$ be the gap.  Then, for any $ v''\in B_{\frac{\varepsilon}{4 n}}(v)$, the inequality holds for $v''$, since $x$ and $q$ are bounded in $[0,1]^n$. Therefore, we verified that $\left\{ v\in U:(Q(v),T(v))\bigcap D\neq \emptyset  \right\}$ is a relatively closed subset of $U$, hence is Borel measurable. Then by then by the Kuratowski–Ryll-Nardzewski measurable selection theorem, there exists a selection $(q^*,t^*)$ from $(Q,T)$ on $U$ that is Borel measurable. Now, extend $(q^*,t^*)$ to $V$ by letting $(q^*,t^*)=(q,t)$ on $V$. Since $(q^*,t^*)$ and $(q,t)$ are Borel measurable on open set $U$ and closed set $V$, respectively, the extension $(q^*,t^*)$ is also Borel measurable. It is easy to verify that $(q^*,t^*)$ satisfies $(IC)$ and $(IR)$ for the entire domain $ \R^n_+$. \end{proof}

\subsection{Construction of saddle point in \Cref{thm:info}} \label{ssec:saddle}
\begin{proof}  Fix any $z\in \mathbb{R}_+$. Consider a worst-case value distribution takes the form constructed in \Cref{sec:moment}.  Specifically, the item values $\bm{V}(X)$ are given by: 
\begin{align*}
		V_i(X)=\min\left\{ \alpha X,\alpha e^{\frac{m-\alpha}{\alpha}} \right\}\cdot \frac{m_i}{m},
	\end{align*}
	for all $i$, where $m_i=\E_{G}[v_i]$ and $m=\sum m_i$, and $X$ is distributed according to  $ H(x)=1-\frac{1}{x}$.  (This is precisely the same before, except that $\beta= \alpha e^{\frac{m-\alpha}{\alpha}}$ has been substituted here.)  Let $F_{\alpha}$ denote  the distribution of random variable $\bm{V}(X)$. As before, $\alpha$ must be such that the constraint holds relative to the test functions $\phi_{z}$; namely, $ \E_{F_{\alpha}}[\phi_{z}(\bm{v})]\le\E_{G}[\phi_{z}(\bm{v})]$. The following argument 
slightly modifies the argument in \Cref{thm:moment} and identifies the worst cast distribution $F_{\alpha}$:

	Define the mechanism $M_{\alpha}$ as a pure bundling mechanism with random price distributed according to cdf:
\begin{align*}
	\gamma(p)=\frac{\log(p)-\log(\alpha)}{\log(z)-\log(\alpha)}
\end{align*}
for $p\in[\alpha,z]$. Then, the revenue function
\begin{align*}
	t^*(\bm{v})=&\int_{\alpha}^{\sum v_i} p \gamma(\d p)\\
	=&\frac{\min\left\{ \sum v_i,z \right\}-\alpha}{\log(z)-\log(\alpha)}.
\end{align*}
Hence, $R(M_{\alpha},F_{\alpha})=\frac{1}{\log(z)-\log(\alpha)}\left( \alpha\int_{\alpha}^{z}\frac{x-\alpha}{x^1}\d x+\alpha\frac{z-\alpha}{z} \right)=\sum \alpha$. $\forall M\in\M$:
\begin{align*}
	R(M,F_{\alpha})\le&\sup_{\psi(\cdot)} \int\psi(x)\cdot \left( \bm{V}(x)-\bm{V}'(x)\cdot x \right)\frac{1}{x^2}\d x\\
	=&\sup_{\bm{\gamma}} \sum_i\gamma_i\cdot\frac{m_i}{m}\cdot\alpha\le\sum\alpha.
\end{align*}
$\forall F\in \F_{\bm{z}^*}$:
\begin{align*}
	R(M_{\bm{\alpha}},F)=&\int t^*(\bm{v})F(\d \bm{v})\\
	\ge&\int\frac{\min(\sum v_i,z)-\alpha}{\log(z)-\log(\alpha)}F(\d \bm{v})\\
	=&\int\frac{-\phi_{z}(\sum v_i)+z-\alpha}{\log(z)-\log(\alpha)}F(\d \bm{v})\\
	\ge&\frac{-\E_{\Gamma}[\phi_{z}(\bm{v})]+z-\alpha}{\log(z)-\log(\alpha)}\\
	=&\alpha.
\end{align*}
The last equality is from $\E_{\Gamma}[\phi_{z}(\bm{v})]=\theta(\alpha,z)$. Therefore, $(M_{\alpha},F_{\alpha})$ forms a saddle point under ambiguity set $\F_{z}$.
 
We now vary $z$ and correspondingly vary $\alpha$ so that $F_{\alpha}$ is the worst case distribution within $\bigcap\F_{z}$.  Write
	\begin{align}
	\E_{F_{\alpha}}[\phi_{z}(\bm{v})]= &
	\begin{dcases}
		0&\text{when }z\le \alpha\\
		\textstyle \int_1^{\frac{z}{\alpha}}(z-\alpha x)\frac{1}{x^2}\d x&\text{when }z\in(\alpha,\alpha e^{\frac{m-\alpha}{\alpha}})\\
		z-m&\text{ when }z\ge \alpha e^{\frac{m-\alpha}{\alpha}}.
	\end{dcases}\label{eqn:theta:K}
\end{align}
Note the function $\theta_K(\alpha,z):=	\E_{F_{\alpha}}[\phi_{z}(\bm{v})]$   is jointly continuous in $(\alpha,z)$ and strictly decreasing in $\alpha$ when $z\in(\alpha,\alpha e^{\frac{m-\alpha}{\alpha}})$. Consider the set 
$$A:=\left\{ \alpha\in(0,m]:  \ \theta(a,z)\le \E_{G}[\phi_z(\bm{v})], \forall z\in\R_+ \right\}.$$  $A$ is non-empty since $m\in A$. Let $\alpha^*=\inf A$. Then by the continuity of $\theta$, $\forall z$, $\theta(\alpha^*,z)\le \E_{G}[\phi_z(\bm{v})]$. Equality must hold at some $z^*\in(0,\infty)$, because otherwise in the region $[\alpha^*,\alpha^{*} e^{\frac{m-\alpha^*}{\alpha^*}}]$ (where $\theta$ depends on $\alpha$), $\E_{G}[\phi_{z}(\bm{v})]$ is bounded away from $\theta(\alpha^*,z)$ and $\alpha^*$ could have been chosen strictly smaller, contradicting $\alpha^*=\inf A$. \footnote{The parameters $(\alpha^*,z^*)$ are exactly $(\pi^*, \overbar{s}^*)$ in Theorem 3 of \cite{deb-roesler2021}. Such connection is observed in \cite{du2018commonvalue} Proposition 1 and \cite{ravid2019learning} Lemma 3 for Pareto distribution with tail index 1, and \cite{kartik2020lemon} Proposition 2 for Pareto distributions with general tail indices.}   

Now, let $F^*:=F_{\alpha^*}$ and $M^*:= M_{\alpha^*}$.  By construction, $F^*$ belongs to $\F_{z}$ for all $z\in \mathbb{R}_+$.  Since $F^*$ is the worst-case distribution given the ambiguity set $\F_{z^*}$, by \Cref{cor:beyond}, $F^*$ is the worst case distribution  and $M^*$ is robustly optimal, given the  ambiguity set $\bigcap_{z}\F_{z}\subset \F_{z^*}$. 
\end{proof}

\subsection{Proof of the statement in \Cref{fn:lm} } \label{sec:add-lemma}

\begin{proposition}
	Let $\mathcal{F}$ be defined as in \Cref{eq:ambiguity-partial}. Suppose further $\Omega$  is downward closed; i.e.,
	\begin{align}
		(\bm{m},\bm{s}) \in \Omega \,\mbox{ and } \, \bm{0}\le \bm{s'} <\bm{s} \implies (\bm{m},\bm{s}')\in \Omega.\label{eq:cone}
	\end{align}
Then, $\mathcal{F}$ is regular.
	\label{lem:continuous:lim1}
\end{proposition}

\begin{remark} \label{rem:upper-bound}
	We argue in words that imposing \cref{eq:cone} on $\Omega$ is without loss of generality. This is because in the max-min problem, nature will never choose  any $F'\in \F$ with $(\sigma_K(F'))_K=\bm{s}'$  if there exists $F\in \F$ such that $( \sigma_K(F))_K=\bm{s}>\bm{s}'$.  For each marginal $K$, nature can always choose a mean-preserving spread, say $F''$, of $F'$  with a very large dispersion such that the probability of a zero value is close to one. Hence for this hypothetical distribution, the revenue can be arbitrarily close to $0$. Mixing $F''$ with $F'$ with sufficiently small probability strictly reduces revenue while keeping the moments below $\bm{s}$.
	\label{rmk:inequality}
\end{remark}

\begin{proof}
	First, $\F$ is clearly convex since all the constraints are linear in $F$ and both $\Omega$ and is convex.
	
	Next, we prove that it is closed under the weak topology.  Consider any sequence $\left\{ F_n \right\}\subset \F$ and $F_n\xrightarrow{w}F$. Let $m_i^n=\mu_i(F_n)$ and $s_K^n=\sigma_K(F_n)$. Without loss of generality, we pick a subsequence that $\lim_{n\to \infty}m_i^n=m_i$ and $\lim_{n\to \infty} s_K^n=s_K$.\par
  Let $\overline s_K:=\sup_{n\in\mathbb{Z}_+} s^n_K$, for each $K\in \K$.
	Next, apply \Cref{lem:port} below with $\overbar{h}(\bm{v})={\phi}_K(\sum_{j\in K}v_j)$, $h(\bm{v})=v_i$, and $C=\overline s_K$. Since $\Omega$ is compact, $\overline s_K$ is finite. \Cref{lem:port} implies that:
	\begin{align*}
		\E_{{F}} [v_i]=&\lim_{n\to \infty} \E_{{F}_n}[v_i]\\
		=&\lim_{n\to \infty} \mu_i({F}_n)=m_i.
		\end{align*} 
	Since $\Omega$ is compact, we must have $(\mu_i(F))=\bm{m}\in \Omega$.\par
	
	Now we verify that $(\sigma_K(F))\in \S$.
	\begin{align*}
	\sigma_K(F)=&\E_F[\phi_K(\tsum_{j\in K}(v_j))]\\
	\le&\liminf_{n\to \infty}\E_{{F}_n}\left[ {\phi}_K(\tsum_{j\in K}v_i) \right]\\
	=&\liminf_{n\to \infty}s^n_K\\
	=&s_K.
	\end{align*}
	The inequality is implied by ${\phi}_K$ being continuous and bounded below (Portmanteau theorem). Therefore, since $\S$ is compact and satisfies \cref{eq:cone},  $(\sigma_K(F))\le\bm{s}$ implies that it is in $\S$.\par
	
	Finally, we prove that $\F$ is tight.  Since $\Omega$ is bounded, there exists $L>0$ such that $\sum_{i\in N} m_i<L$ for all $\bm{m}\in \Omega$. For each $k>0$, consider a set 
	$$U(k):=\left\{\bm{v}\in \R_+^n:  \tsum_{i=1}^n v_i\le k \right\}.$$	
	The set $U(k)$ is compact.  By Markov's inequality, we have
	$$\P_{F} \left\{\bm{v} \not\in U(k)\right\} \le \frac{\sum_{i=1}^n \E_F[v_i]}{k}< \frac{L}{k},$$
for all $F\in \F$.  Hence, for any $\varepsilon>0$, one can take $k$ large enough so that 
 $\P_{F} \left\{\bm{v} \not\in U(k)\right\}<\varepsilon$, as was to be shown.	
\end{proof}

\begin{lemma} \label{lem:port}
	$\forall l\in\mathbb{N}$, let $\left\{ F_n \right\}\subset \Delta(\R^l)$, $F_n\xrightarrow{w}F$, and $\overbar{h}\in \mathcal C(\R^l)$ is a nonnegative function. If  $\int \overbar{h}(x) F_n(\d x)\le C $ for all $n$, then for all nonnegative  function $h\in \mathcal C(\R^l)$ such that  $\lim\limits_{|x|\to \infty}\left|\frac{h(x)}{\overbar{h}(x)}\right|=0$, 
	\begin{align*}
	\int h(x)F(\d x)=\lim_{n\to \infty} \int h(x)F_n(\d x).
	\end{align*}
\end{lemma}

\begin{proof}
	First, the Portmanteau theorem implies $\int h(x)F(\d x)\le \liminf_{n\to\infty} \int h(x) F_n(\d x)$ for any nonnegative function $h$. Hence, it suffices to prove that $\int h(x) F(\d x)\ge \limsup_{n\to\infty}$ $\int h(x) F_n(\d x)$. 
	Suppose  for the sake of contradiction that $\int h(x) F(\d x)< \limsup_{n\to\infty} \int h(x) F_n(\d x)$ (which is bounded by $C$). Without loss of generality,  we pick a subsequence such that  $\int \overbar h(x) F_n(\d x)$ converges. Along that subsequence, 
	\begin{align*}
	\int h(x)F(\d x)<&\limsup_{n\to \infty} 	\int h(x)F_n(\d x)\\
	\iff\ \exists A>0\ s.t.\ \int [\overbar{h}(x)-Ah(x)]F(\d x)>&\liminf_{n\to \infty} \int [\overbar{h}(x)-Ah(x)]F_n(\d x)\\
	\iff\ \int [\overbar{h}(x)-Ah(x)+B] F(\d x)>&\liminf_{n\to \infty} \int [\overbar{h}(x)-Ah(x)+B ]F_n(\d x),
	\end{align*} 
	for any $B\in \R$. Since $\lim\limits_{|x|\to \infty}\left|\frac{h(x)}{\overbar{h}(x)}\right|=0$, there exists sufficiently large $B$ such that  $\overbar{h}(x)-Ah(x)+B\ge0$. The last part contradicts  the Portmanteau theorem:  $\int [\overbar{h}(x)-Ah(x)+B] F(\d x)\le\liminf_{n\to \infty} \int [\overbar{h}(x)-Ah(x)+B ]F_n(\d x)$.
\end{proof}

\subsection{ Ambiguity Sets with Domain Restrictions} \label{app:domain}

Here, we consider the ambiguity set $\F$ which satisfies mean conditions (a special case of $S$ being a singleton) but must satisfy domain restrictions  instead of dispersion moment conditions.  Specifically, fix any partition $\K$ of $N$, with each element $K\in \K$ interpreted as a bundle of goods. 

The seller now knows that the buyers' values lie within the domain  $D:=\big\{ \bm{v}\in \R^n_+|\forall K\in\K, \sum_{i\in K} v_i\in[0,\overbar{v}_K] \big\}$. The ambiguity set is now:
\begin{align*}
	\F=\left\{ F\in \Delta(D): \, \E_F[v_i]=m_i, \forall i\right\},
\end{align*}
where $0< \sum_{i \in K}m_i<\overbar{v}_K$ for each $K\in \K$.

 As before, we exhibit a saddle point $(M^*,F^*)\in(\M,\F)$ and prove that it satisfies the requirement \Cref{eqn:1}.\footnote{Just like the dispersion moment conditions, it is easy to see that the ambiguity set $\F$ exhibits $\K$-Knightian ambiguity as defined in \Cref{sec:general}. Therefore, \Cref{thm:partial:bundling} suggests that $\K$-bundled sales is robustly optimal. (Regularity is easy to verify in this case.)}

\paragraph{Construction of $F^*$.} Let the support of $F^*$ be defined as a parametric curve $\bm{V}(x): [1,\infty)\to D$ with the value of item $i$ given by:
\begin{align*}
	V_i(s)=\min\left\{ \alpha_{K(i)}\cdot x, \overbar{v}_{K(i)}\right\}\cdot \frac{m_i}{\sum_{j\in K(i)}m_j},
\end{align*}
where $s$ is a scalar distributed from $[1,\infty)$ according to cdf $H$:
\begin{align*}
	\mathrm{Prob}(x\le y)=H(y)=1-\frac{1}{y},
\end{align*}
and $0<\alpha_K<\sum_{j\in K}m_j $ satisfy:
\begin{align}
	\alpha_K(1+\log(\overbar{v}_K/\alpha_K))=\tsum_{j\in K}m_j.\label{eqn:mi:domain}
\end{align}
The choice of $(\alpha_K)$ guarantees that the mean conditions are satisfied.\par

\paragraph{Construction of $M^*$.} The construction of $M^*$ is exactly the same as in \Cref{sec:moment}. The seller sells each bundle $K$ separately at independent random prices distributed according to $G_K$:
\begin{align*}
\begin{dcases}
	q^*_i(\bm{v})=G_{K(i)}\big(\tsum_{j\in K(i)} v_j\big),\\
	t^*(\bm{v})=\tsum_{K\in \K}\int_{p\le \sum_{j\in K} v_j}pG_K(\d p).
\end{dcases}
\end{align*}
The cdf $G_K$ is defined via the density function:
\begin{align*}
	g_K(v):=\frac{1}{\log(\overbar{v}_K/\alpha_K)v}
\end{align*}
on $[\alpha_K, \overbar{v}_K]$ and zero elsewhere.

\begin{theorem}	\label{thm:domain}
	The pair $(M^*,F^*)$ is a saddle point  satisfying (\ref{eqn:1}).  In the saddle point,  seller attains revenue $\sum_{K\in\K} \alpha_K$ by selling each bundle $K$ separately at a random price according to $G_K$.
\end{theorem}
\begin{proof}. We first compute the value $R(M^*,F^*)$. 	On the support of $F^*$,
	\begin{align*}
		t^*(\bm{v})=&\sum_{K\in\K}\frac{\sum_{j\in K}v_j-\alpha_K}{\log(\overbar{v}_K/\alpha_K)}.
		\end{align*}
		Hence,
		\begin{align} R(M^*,F^*)=&\int t^*(\bm{v})F^*(\d \bm{v})\cr
			=&\sum_{K\in\K}\frac{\sum_{j\in K}m_j-\alpha_K}{\log(\overbar{v}_K/\alpha_K)} \cr
			=&\sum_{K\in \K} \alpha_K.\label{eq:saddle:domain}
	\end{align}

	Next, we show that $M^*\in\arg\max_{M\in \M} R(M,F^*)$. To this end, fix any $ M=(q,t)\in\M$. Since the support of $F^*$ is a parametric curve $\bm{V}(x)$, the mechanism $M$ can be represented equivalently via $(\psi(x),\tau(x)):=(q(\bm{V}(x)),t(\bm{V}(x)))$.  Since $M$ satisfies $(IC)$, it must satisfy the envelope condition:
	\begin{align*}
	\tau(x)=&\psi(x)\cdot \bm{V}(x)-\int_1^x \psi(z)\cdot \bm{V}'(z)\d z.
	\end{align*}
	Hence,
		\begin{align}
		 R(M,F^*)\le&\sup_{\psi} \int \psi\cdot\left( \bm{V}(x)-\bm{V}'(x)\frac{1-H(x)}{h(x)} \right)H(\d x)\cr
		 =&\sup_{\psi}\sum_{i}\int_1^{\frac{\overbar{v}_{K(i)}}{\alpha_{K(i)}}}\psi_i(x)\cdot 0 H(\d x)+\int_{\frac{\overbar{v}_{K(i)}}{\alpha_{K(i)}}}^{\infty} \psi_i(x)\cdot \gamma_i\cdot \overbar{v}_{K(i)} H(\d x)\cr
		 \le&\sum_{i}\gamma_i\cdot\overbar{v}_{K(i)}\cdot\frac{\alpha_{K(i)}}{\overbar{v}_{K(i)}}=\sum_{K\in\K} \alpha_K
	=R(M^*,F^*), \label{eq:saddle1:domain}
	\end{align}
where $\gamma_i:= \frac{m_i}{\sum_{j\in N} m_j}$.  The second inequality is from $\psi_i\le1$. The second equality is from $\sum_{i\in N}\gamma_i=1$. The last equality is from \cref{eq:saddle:domain}.
	
	Finally,  we show that $F^*\in\arg\min_{F\in\F} R(M^*,F)$.   To this end, observe 
	\begin{align*}
		t^*(\bm{v})\ge&\sum_{K\in\K}\frac{\sum_{j\in K}v_j-\alpha_K}{\log(\overbar{v}_K/\alpha_K)}.
		\end{align*}
		To see why the inequality holds, note that $t^*(\bm{v})=RHS$ when $\sum_{j\in K}v_j\in[\alpha_K,\overbar{v}_K]$.  Outside that region,  $t^*(\bm{v})$ is flat whereas the RHS is strictly increasing in $\sum_{j\in K}v_j$ whenever it is below $\alpha_K$. It then follows that 
	\begin{align}
		R(M^*,F)\ge&\int\sum_{K\in\K}\frac{\sum_{j\in K}v_j-\alpha_K}{\log(\overbar{v}_K/\alpha_K)} F(\d \bm{v})\cr
		=&\sum \alpha_K=R(M^*,F^*)  \label{eq:saddle2:domain}.
	\end{align}
	Combining \cref{eq:saddle1:domain} and \cref{eq:saddle2:domain}, the desired result follows.
\end{proof}

\begin{remark}  The saddle point here bears uncanny resemblance to that presented in \Cref{thm:moment}.  In particular, the worst-case distributions $F^*$ are remarkably similar to each other in the two cases.  In fact, they are identical if one were to replace  $\overbar{v}_K$ by $\beta_K$ in the formula. For any $ (\overbar{v}_K)_{K\in \K}>(m_K)_{K\in \K}$, one can find $(\phi_K)_{K\in \K}$ such that the optimal distribution for nature is identical.  In this sense, the dispersion moment conditions play similar roles to  upper bounds of bundle values. Intuitively, facing a dispersion moment for $\sum_{i \in K} v_K$, there is a largest  bundle value $\sum_{i \in K} v_K$ beyond which nature finds it too costly to load any probability mass. At the same time, the resemblance is less than exact for the optimal selling mechanism.  Given the absence of dispersion moment conditions, the revenue function $t^*$ is linear (instead of concave) within the support of $F^*$ (see the figures below).  Despite these differences, the optimal mechanisms $M^*$ are qualitatively similar between  the two cases. 
\end{remark}

\begin{remark}  As with \Cref{sec:moment}, \Cref{thm:domain} specializes to two canonical cases. When $\K$ is the finest partition of $N$, $M^*$ involves a separate sales mechanism, and when $\K$ is the coarsest partition, $M^*$  involves a pure bundling mechanism.  The support of $F^*$ and the revenue from the optimal mechanism are depicted in each of these two cases in  \Cref{fig:F1',fig:T1'} and   \Cref{fig:F3,fig:T3}. Compared with \Cref{sec:moment}, the only differences  are that the revenue functions are linear (rather than concave) within the support of $F^*$, as noted above. 
\end{remark}
 
\begin{figure}[htbp]
	\begin{minipage}[tb]{0.48\textwidth}
		\centering
		\includegraphics[height=5cm]{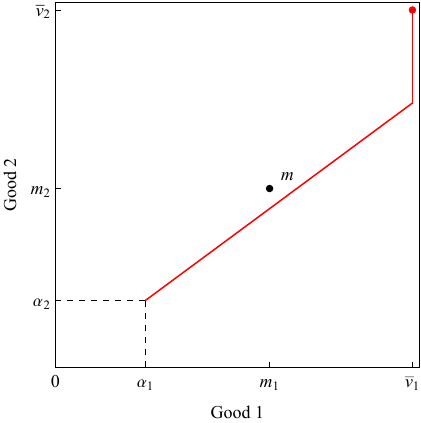}
		\begin{minipage}{0.7\linewidth}
		\caption{\small Valuation distribution $F^*$ when $\K$ is the finest partition}				\label{fig:F1'}
	\end{minipage}
	\end{minipage}
	\begin{minipage}[tb]{0.47\textwidth}
		\centering
		\includegraphics[height=5cm]{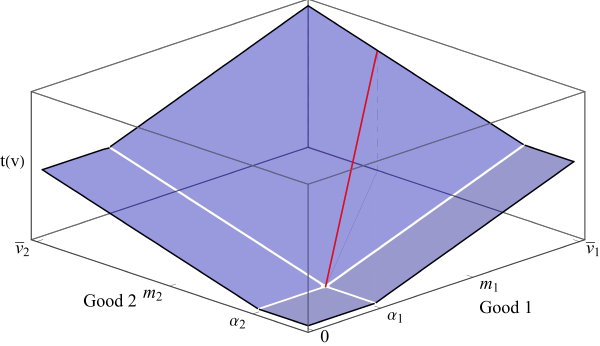}
	\begin{minipage}{0.7\linewidth}
		\caption{\small Revenue from mechanism $M^*$ when $\K$ is the finest partition}				\label{fig:T1'}
	\end{minipage}
	\end{minipage}

\begin{minipage}[H]{0.8\textwidth}
	\centering
	\small Note: $m_1=0.6$, $m_2=0.5$, $\overbar{v}_1=\overbar{v}_2=1$.
\end{minipage}

\end{figure}

\begin{figure}[htbp]
	\begin{minipage}[tb]{0.48\textwidth}
		\centering
		\includegraphics[height=5cm]{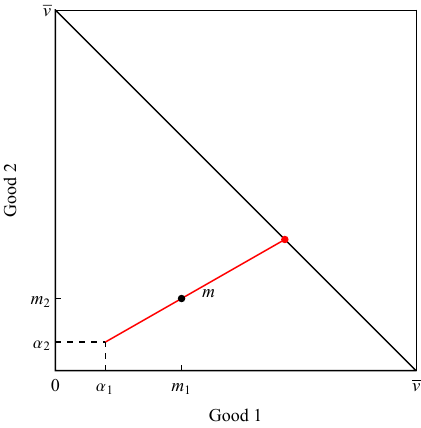}
		\begin{minipage}{0.7\linewidth}
		\caption{\small Valuation distribution $F^*$ when $\K$ is the coarsest partition}				\label{fig:F3}
	\end{minipage}
	\end{minipage}
	\begin{minipage}[tb]{0.47\textwidth}
		\centering
		\includegraphics[height=5cm]{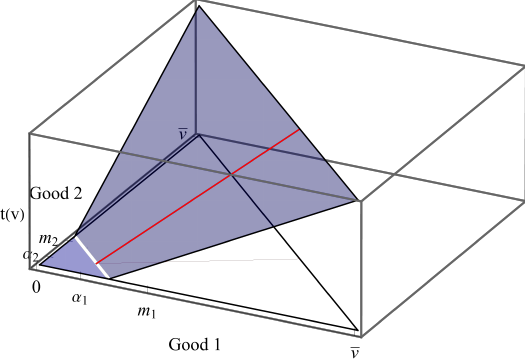}
		\begin{minipage}{0.7\linewidth}
		\caption{\small Revenue from mechanism $M^*$ when $\K$ is the coarsest partition}				\label{fig:T3}
	\end{minipage}
	\end{minipage}

\begin{minipage}[H]{0.8\textwidth}
	\centering
	\small Note: $m_1=0.7$, $m_2=0.4$, $\overbar{v}_{\{1,2\}}=2$.
\end{minipage}
\end{figure}

\end{document}